%% file: main.tex
\documentclass[journal]{IEEEtran}
\IEEEoverridecommandlockouts

\usepackage{graphicx}
\usepackage{cite}
\usepackage{amsfonts} % Required for \mathbb
\usepackage{amsmath}  % Required for various math symbols
\usepackage{bm}
\usepackage{psfrag}
\usepackage{amsbsy}
\usepackage{amssymb}
\usepackage{mathrsfs}
\usepackage{stfloats}
\usepackage{dsfont}
\usepackage{setspace}
\usepackage{array}
\usepackage{bm}
\usepackage{bbm}
\usepackage{mathabx}
\usepackage{algorithmic}
\usepackage{algorithm}
\usepackage{glossaries}
\usepackage{color}
\usepackage{amsthm}
\usepackage{lettrine}

\usepackage[textwidth=30mm]{todonotes}

\title{Energy Efficiency Maximization of MIMO Systems through Reconfigurable Holographic Beamforming}
\author{
	Robert Kuku Fotock, {\em Student Member, IEEE},  Alessio Zappone,\\ {\em Fellow, IEEE}, Agbotiname Lucky Imoize, {\em Senior Member, IEEE}, Marco Di Renzo, {\em Fellow, IEEE}, 
	\thanks{R. K. Fotock is with the University of Cassino and Southern Lazio. A. Zappone is with the University of Cassino and Southern Lazio and with CNIT, Italy (\{robertkuku.fotock, alessio.zappone\}@unicas.it). A. L. Imoize is with CNIT and with Politecnico di Torino, Italy (agbotiname.imoize@polito.it). M. Di Renzo is with Universit\'e Paris-Saclay, CNRS, CentraleSup\'elec, Laboratoire des Signaux et Syst\`emes, France (marco.di-renzo@universite-paris-saclay.fr). 
}}

\input{abbreviations}

\begin{document}

\maketitle

	\begin{abstract}
		
This study considers a point-to-point wireless link, in which both the transmitter and receiver are equipped with multiple antennas. In addition, two reconfigurable metasurfaces are deployed, one in the immediate vicinity of the transmit antenna array, and one in the immediate vicinity of the receive antenna array. The resulting architecture implements a holographic beamforming structure at both the transmitter and receiver. In this scenario, the system energy efficiency is optimized with respect to the transmit covariance matrix, and the reflection matrices of the two metasurfaces. A low-complexity algorithm is developed, which is guaranteed to converge to a first-order optimal point of the energy efficiency maximization problem. Moreover, closed-form expressions are derived for the metasurface matrices in the special case of single-antenna or single-stream transmission. The two metasurfaces are considered to be nearly-passive and subject to global reflection constraints. A numerical performance analysis is conducted to assess the performance of the proposed optimization methods, showing, in particular, that the use of holographic beamforming by metasurfaces can provide significant energy efficiency gains compared to fully digital beamforming architectures, even when the latter achieve substantial multiplexing gains. 

	\end{abstract}

\section{Introduction}
\lettrine[nindent=0.1em,lines=2]{E}{nergy} efficiency (EE) continues to be a major performance requirement of future wireless communications, especially considering that 5G did not achieve the desired 2000x EE increase \cite{DavidEE}, mainly due to the use of a large number of digital antennas, and the resulting high static power consumption \cite{HuaweiReport}. In order to face this challenge, reconfigurable metasurfaces have emerged as one of the main technologies for the sixth generation (6G) of wireless networks, due to their ability to provide a large amount of degrees of freedom with limited power consumption\cite{huang2019reconfigurable,di2020smart}. Reconfigurable metasurfaces are nearly-passive devices, fully operating in the analog domain, i.e. they require no dedicated amplifier and no conversion between the analog and digital domains. By requiring only a limited amount of static energy to enable the reconfiguration of the reflecting elements, metasurfaces can provide high beamforming gains, with a low energy consumption, and their use has been considered in many applications for 6G, analyzing the emerging trends, recent advancements, potential opportunities and challenges that could impact the integration of reconfigurable metasurfaces into 6G wireless networks \cite{strinati2021reconfigurable,zhang2025smart,katwe2024overview}. 

\subsection{Literature review}
Reconfigurable metasurfaces have been first deployed far from the transceivers, as an efficient way of creating new optical paths from transmitters to receivers. In this configuration, they are called reconfigurable intelligent surfaces (RISs), and provide a measure of control of the propagation environment. On the other hand, metasurfaces have been considered also for deployment in the vicinity of the wireless transceivers, with several application that have emerged in the last years:
\begin{itemize}
\item[(a)] a first use of metasurfaces at the transmit side, has been the implementation of \emph{index modulation} techniques. In this context, data can be encoded into the activation pattern of the metasurface elements \cite{BasarRIS}. 
\item[(b)] another example of metasurfaces employed at the transmit side is the use of \emph{dynamic metasurfaces}, which have been proposed as an efficient approach for beam focusing. This has been shown to provide extremely high localization accuracies and efficient energy usage in multiple-antenna communications \cite{Yang2025,Zhang2022}. 
\item[(c)] A third application of metasurfaces deployed at wireless transceiver is their use to implement holographic MIMO (HMIMO) \cite{huang2020holographic}. 
\end{itemize}
HMIMO is traditionally defined as a dense panel of radiating elements spaced at subwavelength distances, which are connected to one or few transmit digital chains.  \cite{pizzo2020spatially}. Usually, the panel is implemented through a reconfigurable metasurface, which is usually referred to as a reconfigurable holographic surface (RHS) \cite{Deng21,gong2023holographic}, which is wired to the feeders.  The advantage of beamforming through RHSs is that, compared to active antenna arrays, RHSs have a lower energy expenditure and cost \cite{zeng2022reconfigurable}, and, thus, can be equipped with a larger number of antenna elements~\cite{you2022energy,guo2023green}. RHSs are being considered an energy-efficient evolution of the all-digital MIMO technology, and have been demonstrated to achieve low-cost and energy-efficient wireless communication and sensing~\cite{deng2023reconfigurable,an2023tutorial,di2025reconfigurable}. As mentioned, RHSs are usually wired to the transceiver feeders, i.e. they are an integral part of the transceiver. In this configuration, the transmit signal propagates to the RHS through a waveguide and is then reflected to the receiver. In~\cite{wei2022multi}, the downlink of multi-user MIMO communications is studied. The work extends an EM-compliant channel model to multi-user setups expressed in the wavenumber domain employing the Fourier plane wave approximation. In \cite{li2024energy}, a BS equipped with a switch-controlled RHS-aided beamforming architecture is considered, and the EE maximization problem is tackled via alternating optimization of the holographic beamformer, the digital beamformer, and the transmit power. A near-field channel model was proposed in~\cite{gong2024holographic}, and the capacity limit of a point-to-point MIMO system with holographic beamforming is investigated. In  \cite{AnSIM2023}, the use of stacked intelligent metasurfaces (SIMs) has been explored for holographic beamforming. In this context, a MIMO system is considered, and the SIM is designed to establish a desired equivalent MIMO channel. In \cite{Wan21}, RHSs are considered for THz-based networks, developing a model and analyzing the resulting performance. In~\cite{joy2024reconfigurable}, an electromagnetic framework for designing a holographic surface is developed. The performance and power consumption of the framework were compared to those of passive metasurfaces and MIMO digital antenna arrays.  In~\cite{qian2024spectral}, channel models for holographic MIMO communications are developed, and the system spectral efficiency is analyzed. The use of an RHS has been considered also in networks aided by uncrewed aerial vehicles (UAVs), as a way of aiding the communications and performing energy harvesting to power the UAV \cite{Song25}. Artificial intelligence has also been considered for power allocation in cell-free networks employing holographic beamforming  \cite{adhikary2024power}. The study in~\cite{bahanshal2024holographic} focuses on the number of antennas required to achieve energy efficiency, in a wireless link employing holographic beamforming by RHSs. In \cite{Jalali25}, the weighted sum-rate maximization problem is tackled in the downlink of a multi-user network in which a BS with a uniform linear array serves single-antenna users through multiple WsRHSs.

While all previous works consider RHSs that are wired to the transceiver feeders, another possible architecture is to employ a wireless connection between the transceiver feeders and the metasurfaces. In this configuration, the RHS is simply deployed in the vicinity of the transceiver to reflect/refract the signal transmitted by the digital antennas, but is not wired to the transceiver hardware. Instead, the signal simply propagates via wireless connection from the transmit antennas to the RHS or from the RHS to the receive antennas. The peculiarity of this setup is that the RHS is deployed closer to the transceiver antenna array than the Fraunhofer's distance, and thus a spherical wave propagation model must be considered, in place of the traditional plane wave propagation model. This architecture was considered in \cite{Zheng21,You2021,Mei2022Multi,zappone2022energy,zeng2023reconfigurable,Interdonato24,Mishra24} and has been referred to as holographic RIS \cite{zappone2022energy}, reconfigurable intelligent BS \cite{Interdonato24,Mishra24}, or reconfigurable refractive surfaces \cite{zeng2023reconfigurable}. In the following, in order to distinguish it from the more traditional, wired RHS, we will use the term Wireless RHS (WsRHS).  In \cite{Zheng21}, two WsRHSs are deployed at the transmitter and receiver, respectively, and the minimum signal-to-interference-plus-noise-ratio (SINR) in a multi-user MISO system is maximized with respect to the transmit beamforming and metasurface coefficients. An alternating maximization algorithm with respect to each metasurface reflection matrix and transmit beamforming is developed. In \cite{You2021}, the received power is maximized in a system in which two WsRHSs are deployed to assist the communication between a single-antenna transmitter and a single-antenna receiver. Suboptimal optimization techniques are developed based on alternating optimization or assuming that the channel between the two metasurfaces has rank one. A similar model is considered in \cite{zappone2022energy}, but the focus of the paper is on the maximization of the EE. A novel method based on sequential programming is developed, which is shown to perform similarly as alternating optimization, but with lower computational complexity. In \cite{Mei2022Multi}, a network aided by multiple metasurfaces is considered, and optimal metasurface selection and beam routing are performed. In~\cite{zeng2023reconfigurable}, a single WsRHS is placed in the near-field of the transmitter, and its size is optimized for EE maximization.

\subsection{Contributions}
From the above literature review, it emerges that the vast majority of the previous works on HMIMO and RHSs focus on the analysis or optimization of the network achievable rate and throughput, neglecting the system EE. Moreover, all previously cited works assume that multiple antenna are deployed only at the base station, while single-antenna mobile users are considered. EE optimization in RHS-aided MIMO networks is still an open problem, even for single-link systems. Indeed, only a few studies have focused on EE optimization in MIMO systems aided by metasurfaces deployed far from the transceivers. In \cite{ZapTWC2021},  upper- and lower-bounds on the EE of a MIMO link are optimized assuming a single-stream transmission. In \cite{You2021b}, the trade-off between EE and spectral efficiency is studied, employing the weighted minimum mean squared error method to tackle the problem. In \cite{Sharma2022}, the EE of a MIMO link employing simultaneous wireless information and power transfer (SWIPT) is considered. However, \cite{Sharma2022} does not optimize the EE, which is defined as the ratio between rate and power consumption, addressing instead the simpler problem of minimizing the difference between rate and power. In \cite{Soleymani24}, a MIMO network with finite block-length transmissions is considered, and the network EE is optimized through sequential fractional programming. At present, no work considers the optimization of the EE in an RHS-based MIMO system. This work aims at closing this gap, considering a MIMO link with multiple antennas equipped at both the transmitter and receiver, and two WsRHSs are deployed in the near-field region of the transmitter and receiver antenna arrays, respectively. Both WsRHS are not wired to the transceiver hardware and are considered to be nearly-passive devices, i.e. no amplifiers are equipped on the WsRHSs. The transmit signal propagates from the digital transmit antenna array to the transmit WsRHS, which reflects it to the receive WsRHS, which, finally, reflects it to the receive antenna array. In this context, the following contributions are made:
\begin{itemize}
\item A novel algorithm is developed to optimize the reflection matrices of the two WsRHSs and the transmit covariance matrix of the transmitter for EE maximization. The proposed algorithm leverages the sequential fractional programming method, coupled with a new reformulation of the optimization problem that deals with unit-rank constraints without resorting to the semidefinite relaxation method. As a result, the proposed method monotonically improves the EE value and converges to a first-order optimal point of the EE maximization problem. 
\item In the special case of a single transmit and receive antennas, the optimal WsRHS matrices and transmit power are found in closed-form, which dispenses with the need to running iterative algorithms and allows for a closed-form analysis of the performance of the optimized system. Closed-form expressions of the optimal WsRHS matrices are also found in the special case in which multiple antennas are employed at both the transmitter and receiver, but a single data stream is employed. 
\item The problems are solved subject to a power constraint at the transmitter and global reflection constraints at the two WsRHSs. Metasurfaces with global reflection constraints are a recent type of metasurfaces which generalize traditional metasurfaces with local reflection constraints \cite{DiRenzoGlobal}. Specifically, while in metasurface with local reflection constraints each reflection coefficient is separately constrained to have modulus not greater than one, in a metasurface with global reflection constraints a single reflection constraint is enforced on all of the reflection coefficients, requiring that the total power reflected by the metasurface is not greater than the total power that impinges on it.
\item Numerical results are provided to assess the performance of the proposed system, in terms of EE and capacity. Interestingly, the analysis reveals that the use of WsRHSs allows a significant reduction of the number of the digital antennas deployed at the transmitter and receiver, leading to significant EE gains and satisfactory capacity levels. 
\end{itemize}
The rest of the paper is organized as follows. Section \ref{Sec:SysModel} introduces the system model and formulates the EE maximization problem. Section \ref{Sec:SISO} solves the EE maximization problem considering that both the transmitter and the receiver are equipped with a single antenna. Section \ref{Sec:MIMO_ST} and \ref{Sec:MIMO_MT} develop resource allocation algorithms for the multiple-antenna case, considering single-stream and multi-stream transmissions, respectively. Section \ref{Sec:NUM_ANA} presents the numerical analysis, while Section \ref{Sec:Concl} concludes the paper.

\textbf{Notation}: Scalars, column vectors, and matrices are denoted by lowercase, boldface lowercase, and boldface uppercase letters, e.g. $x$, $\mathbf{x}$, $\mathbf{X}$.  $\text{diag}(\cdot)$ denotes a diagonal matrix, $\tr(\mathbf{S})$ and $\mathbf{S}^{-1}$ denote the trace and inverse of a square matrix. $\mathbf{S} \succeq \mathbf{0}$ indicates that $\mathbf{S}$ is positive semi-definite. $\mathbf{I}_M$ denotes an identity matrix of size $M \times M$. The distribution of a circularly symmetric complex Gaussian random variable $x$ with a mean of $\mu$ and a variance of $\sigma^2$ is represented as $x \sim \mathcal{CN}(\mu, \sigma^2)$.

\section{System Model}\label{Sec:SysModel}
Let us consider a MIMO link with $N_{T}$ transmit antennas and $N_{R}$ receive antennas arranged in a rectangular array. A reconfigurable metasurface with $M_{T}$ elements and reflection matrix $\bGamma_{T}=\text{diag}(\gamma_{T,1},\ldots,\gamma_{T,M_T})=\text{diag}(\bgamma_{T})$
is deployed in the immediate vicinity of the transmitter and henceforth called transmit WsRHS. Similarly, another reconfigurable metasurface with $M_R$ elements and reflection matrix $\bGamma_{R}=\text{diag}(\gamma_{R,1},\ldots,\gamma_{R,M_{R}})=\text{diag}(\bgamma_{R})$ is deployed in the immediate vicinity of the receiver and henceforth called receive WsRHS. We assume that both WsRHSs are nearly-passive, i.e. they are not equipped with any radio-frequency amplifier.  

Then, let us denote by $\bC$, $\bH$, and $\bG$ the $M_{R}\times M_{T}$ channel between the two WsRHSs, the $M_T\times N_T$ channel between the transmitter and the WsRHS placed in its vicinity, and the $N_{R}\times M_{R}$ channel between the receiver and the WsRHS placed in its vicinity. Finally, let $\bs$ be the $N_{T}\times 1$ vector of transmit information symbols, with covariance matrix $\bQ=\mathbb{E}[\bs\bs^{H}]$ and subject to the power constraint $\tr(\bQ)\leq P_{max}$, wherein $P_{max}$ is the maximum transmit power of the transmit amplifier. The system model is depicted in Fig.~\ref{fig:systmodel}.

\begin{figure}[!h]
	\centering
	\includegraphics[width=0.3\textwidth]{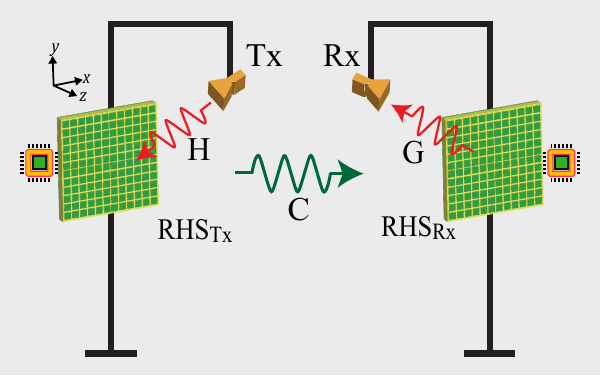}\caption{Holographic MIMO system by WsRHSs} \label{fig:systmodel}
\end{figure}

\subsection{Channel model}
In the following, different channel models are considered for the channels $\bH$ and $\bG$, which are near-field channels, and the channel $\bC$, which is a far-field channel. As for the channels $\bH$, and $\bG$, each entry of the matrices $\bH$ and $\bG$ follow the (deterministic) spherical wave equation, i.e. the $(n,m)$ element of $\bH$, with $n=1,\ldots,M_{T}$ and $m=1,\ldots,N_{T}$, is
\begin{align}
H(n,m)&=\frac{\lambda}{4\pi}\sqrt{\alpha_{n,m}^{RHS_T}\alpha_{n,m}^{Tx}}\frac{e^{-j(2\pi/\lambda)\|\br_{n}^{RHS_T}-\br_{m}^{Tx}\|}}{\|\br^{RHS_T}-\br^{Tx}\|}
\end{align}
wherein $\br_{n}^{RHS_T}$ and $\br_{m}^{Rx}$ are the vectors defining the 3D positions of the $n$-th element of the transmit WsRHS and $m$-th BS antenna, $ \alpha_{m,n}^{Tx}$ denotes the transmit gain of the $m$-th transmit antenna towards the $n$-the element of the transmit WsRHS, and $\alpha_{m,n}^{RHS_T}$ denotes the receive gain of the $n$-th element of the transmit WsRHS, from the $m$-th transmit antenna. The gains $ \alpha_{m,n}^{BS}$ and  $\alpha_{m,n}^{RHS}$ are expressed as 
\begin{align}
\alpha_{n,m}^{Tx}&=\frac{4\pi}{\lambda^{2}}\Delta_{h,Tx}\Delta_{v,Tx}\rho_{n,m}^{Tx}\\
\alpha_{n,m}^{RHS_T}&=\frac{4\pi}{\lambda^{2}}\Delta_{h,RHS_T} \Delta_{v,RHS_T}\rho_{n,m}^{RHS_T}\;,
\end{align}
with $\Delta_{h,Tx}$ and $\Delta_{v,Tx}$ the horizontal and vertical spacing of the transmit antennas, $\Delta_{h,RHS_T}$ and $\Delta_{v,RHS_T}$ the horizontal and vertical spacing of the elements of the transmit WsRHS, $\rho_{n,m}^{Tx}$ and $\rho_{n,m}^{RHS_T}$ the standard directivity factors of the transmitter and transmit WsRHS, respectively \cite{Esposti22,Feng24}. Similarly, the $(n,m)$ element of $\bG$, with $n=1,\ldots,N_{R}$ and $m=1,\ldots,M_{R}$ is expressed as
\begin{align}
G(n,m)&=\frac{\lambda}{4\pi}\sqrt{\alpha_{n,m}^{RHS_R}\alpha_{n,m}^{Rx}}\frac{e^{-j(2\pi/\lambda)\|\br_{n}^{Rx}-\br_{m}^{RHS_R}\|}}{\|\br_{n}^{Rx}-\br_{m}^{RHS_R}\|}\;,
\end{align}
wherein $\br_{n}^{RHS_R}$ and $\br_{m}^{Rx}$ are the vectors defining the 3D positions of the $n$-th element of the receive WsRHS and $m$-th BS antenna, $\alpha_{m,n}^{Rx}$ denotes the receive gain of the $n$-th receive antenna from the $m$-th element of the receive WsRHS, and $\alpha_{m,n}^{RHS,R}$ denotes the transmit gain of the $m$-th element of the receive WsRHS towards the $n$-th receive antenna. The gains $ \alpha_{m,n}^{Rx}$ and $\alpha_{m,n}^{RHS_R}$ are expressed as 
\begin{align}
\alpha_{n,m}^{Rx}&=\frac{4\pi}{\lambda^{2}}\Delta_{h,Rx} \Delta_{v,Rx}\rho_{n,m}^{Rx}\\
\alpha_{n,m}^{RHS_R}&=\frac{4\pi}{\lambda^{2}}\Delta_{h,RHS_R} \Delta_{v,RHS_R}\rho_{n,m}^{RHS_R}\;,
\end{align}
with $\Delta_{h,Rx}$ and $\Delta_{v,Rx}$ the horizontal and vertical spacing of the receive antennas, $\Delta_{h,RHS_R}$ and $\Delta_{v,RHS_R}$ the horizontal and vertical spacing of the elements of the receive WsRHS, and $\rho_{n,m}^{Rx}$ and $\rho_{n,m}^{RHS_R}$ the standard directivity factors of the receiver and receive WsRHS. Finally, the two WsRHSs are assumed to be far from each other, so that signal propagation between them can be modeled following the traditional far-field model, subject to block flat fading. In particular, we consider that each entry of $\bC$ is a realization of a Rice random variable with factor $K$, multiplied by the path-loss coefficient 
\beq
\text{PL}=\text{PL}_{0}\left(\frac{d}{d_{0}}\right)^{-\nu}\;,
\eeq
with $\text{PL}_{0}$ the path-loss at the reference distance $d_{0}$, $d$ the distance between the two WsRHSs, and $\nu$ the path-loss exponent.

\subsection{Problem formulation} 
Given the notation defined above, the input and output signals at the transmit WsRHS are
\begin{align}
\by_{in,T}&=\bH\bQ^{1/2}\bs\;,\;\by_{out,T}=\bGamma_{T}\bH\bQ^{1/2}\bs\;,
\end{align}
the input and output signals at the receive WsRHS are 
\begin{align}
\by_{in,R}&=\bC\bGamma_{T}\bH\bQ^{1/2}\bs\;,\;\by_{out,R}=\bGamma_{R}\bC\bGamma_{T}\bH\bQ^{1/2}\bs\;,
\end{align}
while the signal received at the final destination is written as $\by=\bG\bGamma_{R}\bC\bGamma_{T}\bH\bQ^{1/2}\bs+\bn$, with $\bn\sim{\cal CN}(\bzero,\sigma^{2}\bI_{N_{R}})$ the thermal noise at the receiver. Thus, the system capacity can be written as
\begin{equation}
C=B\log_{2}\left|\bI_{N_{R}}\!+\!\frac{1}{\sigma^{2}}\bG\bGamma_{R}\bC\bGamma_{T}\bH\bQ\bH^{H}\bGamma_{T}^{H}\bC^{H}\bGamma_{R}^{H}\bG^{H}\right|\;.
\end{equation}
The input and output power at the transmit and receive WsRHSs are equal to
\begin{align}
P_{in,T}&=\tr(\bH\bQ\bH^{H})\;,\\
P_{out,T}&=\tr(\bGamma_{T}\bH\bQ\bH^{H}\bGamma_{T}^{H})\;,\\
P_{in,R}&=\tr(\bC\bGamma_{T}\bH\bQ\bH^{H}\bGamma_{T}^{H}\bC^{H})\;,\\
P_{out,R}&=\tr(\bGamma_{R}\bC\bGamma_{T}\bH\bQ\bH^{H}\bGamma_{T}^{H}\bC^{H}\bGamma_{R}^{H})\;,
\end{align}
We consider WsRHSs with global reflection constraints, i.e., at each WsRHS, the input power must be lower than the output power, namely 
$P_{out,R}\leq P_{in,R}$ and $P_{out,T}\leq P_{in,T}$. Hence, the two WsRHSs do not consume any radio-frequency power, but only static power for reconfiguring the reflecting elements. Thus, the total system power consumption is\footnote{We consider the typical scenario in which the metasurface static power consumption does not depend on the values of the reflection coefficients.}
\begin{align}
P_{t}&=\mu \tr(\bQ)+M_{T}P_{T}^{(s)}+M_{R}P_{R}^{(s)}+N_{T}P_{T}^{(a)}+N_{R}P_{R}^{(a)}+\notag\\
&+P_{RHS,0}+P_{0}=\mu \tr(\bQ)+P_{c},
\end{align}
wherein $P_{T}^{(s)}$ and $P_{R}^{(s)}$ are the static power consumption of each element of the transmit and receive WsRHS, $P_{T}^{(a)}$ and $P_{R}^{(a)}$ are the static power consumption of each digital radio-frequency chain of the transmitter and receiver, $P_{RHS,0}$ is the rest of the static power consumed by the two WsRHSs, $P_{0}$ is the static power consumed in the rest of the system, while $P_{c}=M_{T}P_{T}^{(s)}+M_{R}P_{R}^{(s)}+N_{T}P_{T}^{(a)}+N_{R}P_{R}^{(a)}+P_{RHS,0}+P_{0}$. Thus, the EE maximization problem can be stated as 
\begin{subequations}\label{Prob:EE_MIMO}
\begin{align}
&\ds\max_{\footnotesize \bQ,\bGamma_{T},\bGamma_{R}}\frac{C(\bGamma_{T},\bGamma_{R},\bQ)}{P_{t}(\bGamma_{T},\bGamma_{R},\bQ)}\label{Prob:aEE_MIMO}\\
&\;\text{s.t.}\;\tr(\bQ)\leq P_{max}\label{Prob:bEE_MIMO}\\
&\;\quad\;\;P_{out,T}(\bQ,\bGamma_{T})-P_{in,T}(\bQ)\leq 0\label{Prob:cEE_MIMO}\\
&\;\quad\;\;P_{out,R}(\bQ,\bGamma_{T},\bGamma_{R})-P_{in,R}(\bQ,\bGamma_{T})\leq 0 \label{Prob:dEE_MIMO}
\end{align}
\end{subequations}
In the following, Problem \eqref{Prob:EE_MIMO} will be tackled in three different scenarios, with increasing degree of complexity. Specifically:

1)  Section \ref{Sec:SISO} addresses Problem \eqref{Prob:EE_MIMO} in the special case in which a single digital antenna is deployed at the transmitter and receiver, i.e.  $N_{R}=N_{T}=1$. For this case, closed-form solutions for $\bGamma_{T}$, $\bGamma_{R}$, and the transmit power are derived. 

2) Section \ref{Sec:MIMO_ST} addresses Problem \eqref{Prob:EE_MIMO} in the special case in which $N_{R}>1$, $N_{T}>1$, and $\text{rank}(\bQ)=1$, i.e. multiple digital antennas are used at both the transmitter and receiver, but a single data-stream transmission is used. For this case, closed-form solutions for $\bGamma_{T}$ and $\bGamma_{R}$ will be derived, whereas an iterative algorithm will be required to optimize the beamforming strategy. 

3) Section \ref{Sec:MIMO_MT} addresses Problem \eqref{Prob:EE_MIMO} in the general case in which $N_{R}>1$, $N_{T}>1$, and $\text{rank}(\bQ)>1$. In this case, an iterative algorithm will be developed to optimize all of the system radio resources. 
 
\section{Single-antenna, single-stream scenario}\label{Sec:SISO}
Assume $N_{R}=N_{T}=1$. Then, $\tr(\bQ)=p$, with $p$ the transmit power subject to the constraint $p\leq P_{max}$. Let us denote by $\bh$ the $M_{T}\times 1$ channel from the transmit antenna to the transmit WsRHS and by $\bg^{H}$ the $1\times M_{R}$ channel from the receive WsRHS to the receive antenna. Then, the EE expression simplifies to
\begin{equation}
\text{EE}=\frac{B\log_{2}\left(1+\ds\frac{p}{\sigma^{2}}\left|\bg^{H}\bGamma_{R}\bC\bGamma_{T}\bh\right|^{2}\right)}{\mu p+P_{c}}\;,
\end{equation}
whereas $P_{in,T}=p\|\bh\|^{2}$, $P_{out,T}=p\|\bGamma_{T}\bh\|^{2}$, $P_{in,R}=p\|\bC\bGamma_{T}\bh\|^{2}$, $P_{out,R}=\|\bGamma_{R}\bC\bGamma_{T}\bh\|^{2}$. Hence, the EE maximization problem becomes
\begin{subequations}\label{Prob:EE_SISO}
\begin{align}
&\ds\max_{\footnotesize p,\bGamma_{T},\bGamma_{R}}\frac{B\log_{2}\left(1+\ds\frac{p}{\sigma^{2}}\left|\bg^{H}\bGamma_{R}\bC\bGamma_{T}\bh\right|^{2}\right)}{\mu p+P_{c,RIS}+P_{c}}\label{Prob:EE_SISOa}\\
&\;\text{s.t.}\;0\leq p\leq P_{max}\label{Prob:EE_SISOb}\\
&\;\quad\;\;\|\bGamma_{T}\bh\|^{2}\leq \|\bh\|^{2}\;,\;\|\bGamma_{R}\bC\bGamma_{T}\bh\|^{2}\leq \|\bC\bGamma_{T}\bh\|^{2}\label{Prob:EE_SISOd}
\end{align}
\end{subequations}
It is useful to observe that, for any fixed $p$, the optimization of $\bGamma_{T}$ and $\bGamma_{R}$ reduces to optimizing the term $\left|\bg^{H}\bGamma_{R}\bC\bGamma_{T}\bh\right|^{2}$, subject to constraints that do not depend on $p$. Thus, the optimal $(\bGamma_{T}$,  $\bGamma_{R})$ will not depend on the optimal $p$ and so the optimization of the pair $(\bGamma_{T}$, $\bGamma_{R})$ can be decoupled from the optimization of the transmit power $p$. These two subproblems will be tackled in the next two subsections. 

\subsection{Optimization of $\bGamma_{T}$ and $\bGamma_{R}$}\label{Sec:RHS_SISO}
For any given $p$, the problem reduces to
\begin{subequations}\label{Prob:EE_SISO_H}
\begin{align}
&\ds\max_{\footnotesize \bGamma_{T},\bGamma_{R}}\left|\bg^{H}\bGamma_{R}\bC\bGamma_{T}\bh\right|^{2}\label{Prob:EE_SISO_Ha}\\
&\;\text{s.t.}\;\|\bGamma_{T}\bh\|^{2}\leq \|\bh\|^{2}\label{Prob:EE_SISO_Hb}\\
&\;\quad\;\;\|\bGamma_{R}\bC\bGamma_{T}\bh\|^{2}\leq \|\bC\bGamma_{T}\bh\|^{2}\label{Prob:EE_SISO_Hc}
\end{align}
\end{subequations}
Defining the new variables $\bx=\bC\bGamma_{T}\bh$ and $\by=\bGamma_{R}\bx$, and assuming that the matrix $\bC$ has a left inverse\footnote{This happens with probability one whenever $M_T\leq M_{R}$}, i.e. $\bC^{+}\bC=\bI_{N_{MT}}$, then Problem \eqref{Prob:EE_SISO_H} can be restated as
\begin{subequations}\label{Prob:EE_SISO_H2}
\begin{align}
&\ds\max_{\footnotesize \bx,\by}\left|\bg^{H}\by\right|^{2}\label{Prob:EE_SISO_H2a}\\
&\;\text{s.t.}\;\|\bC^{+}\bx\|^{2}\leq \|\bh\|^{2}\label{Prob:EE_SISO_H2b}\\
&\;\quad\;\;\|\by\|^{2}\leq \|\bx\|^{2}\label{Prob:EE_SISO_H2c}
\end{align}
\end{subequations}
\begin{proposition}\label{Prop:OptSISO}
The optimal solution of Problem \eqref{Prob:EE_SISO_H2} is 
\begin{align}
\bx&=\|\bh\|\sqrt{\lambda_{max}}\bu_{max}\label{Eq:Opt_x}\\
\by&=\|\bh\|\sqrt{\lambda_{max}}\frac{\bg}{\|\bg\|}\;\label{Eq:Opt_y},
\end{align}
wherein $\lambda_{max}$ is the maximum eigenvalue of $\bC\bC^{H}$ and $\bu_{max}$ the corresponding unit-norm eigenvector. 
\end{proposition}
\begin{proof}
By Cauchy-Schwartz inequality, \eqref{Prob:EE_SISO_H2a} is maximized when $\by$ is aligned with $g$, i.e. $\by=\frac{\|\by\|}{\|\bg\|}\bg$, with $\|\by\|$ still to be determined. To this end, we observe that $\|\by\|$ should be as large as possible, and so, from \eqref{Prob:EE_SISO_H2c}, it must hold $\|\by\|=\|\bx\|$ and we should make $\|\bx\|$ as large as possible. To this end, we rewrite \eqref{Prob:EE_SISO_H2b} as
\begin{equation}
\|\bC^{+}\bx\|^{2}=\|\bx\|^{2}\|\bC^{+}\tilde{\bx}\|^{2}\leq \|\bh\|^{2}\;,
\end{equation}
with $\tilde{\bx}$ the versor of $\bx$. Then, it holds
\begin{equation}\label{Eq:UBx}
\|\bx\|^{2}\leq\frac{\|\bh\|^{2}}{\|\bC^{+}\tilde{\bx}\|^{2}} \leq \|\bh\|^{2}\lambda_{max}\;,
\end{equation}
wherein the last inequality becomes an equality by choosing $\tilde{\bx}=\bu_{max}$. Hence, \eqref{Eq:Opt_x} is obtained and  \eqref{Eq:Opt_y} follows since $\|\by\|=\|\bx\|$ is optimal. 
\end{proof}
Denoting by $\bar{\bx}$ and $\bar{\by}$ the optimal solution of \eqref{Prob:EE_SISO_H2}, and defining $\bgamma_{T}=[\bGamma_{T}(1,1),\ldots, \bGamma_{T}(M_T,M_T)]$, $\bgamma_{R}=[\bGamma_{R}(1,1),\ldots, \bGamma_{R}(M_R,M_R)]$, $\tilde{\bX}=\text{diag}(\bar{\bx})$, $\tilde{\bH}=\text{diag}(\bh)$, the corresponding $\bGamma_{T}$ and $\bGamma_{R}$ are computed observing that $\bGamma_{T}\bh=\tilde{\bH}\bgamma_{T}$ and $\bGamma_{R}\bx=\tilde{\bX}\bgamma_{R}$. As a consequence, recalling that $\bx=\bC\bGamma_{T}\bh$ and $\by=\bGamma_{R}\bx$, it finally holds
\begin{align}
\bar{\bGamma}_{T}&=\text{diag}(\bgamma_{T})=\text{diag}(\tilde{\bH}^{-1}\bC^{+}\bar{\bx})\label{Eq:OptGamma_T}\\
\bar{\bGamma}_{R}&=\text{diag}(\bgamma_{R})=\text{diag}(\tilde{\bX}^{-1}\bar{\by})\label{Eq:OptGamma_R}
\end{align}

Finally, we can also confirm that the optimized $\bGamma_{T}$ and $\bGamma_{R}$ do not depend on the value of the transmit power $p$. Indeed, plugging the optimized $\bGamma_{T}$ and $\bGamma_{R}$ into the EE yields 
\begin{equation}\label{Eq:OptimalEE}
\text{EE}=\frac{B\log_{2}\left(1+\frac{p}{\sigma^{2}}\lambda_{max}\|\bh\|^{2}\|\bg\|^{2}\right)}{\mu p +P_{c}}\;,
\end{equation}
whose optimization is dealt with next.

\subsection{Optimzation of $p$}\label{Sec:TxPower_SISO}
Given \eqref{Eq:OptimalEE}, and neglecting the inessential (as far as the optimization of $p$ is concerned) constant $B$, the optimal $p$ is found by solving
\begin{subequations}\label{Prob:EE_SISO_P}
\begin{align}
&\ds\max_{\footnotesize p}\frac{\log_{2}\left(1+ a p\right)}{\mu p+P_{c}}\label{Prob:EE_SISO_Pa}\\
&\;\text{s.t.}\; 0\leq p\leq P_{max}\label{Prob:EE_SISO_Pb}
\end{align}
\end{subequations}
wherein we have defined  $a=|\bg^{H}\bar{\bGamma}_{R}\bC\bar{\bGamma}_{T}\bh|^{2}/\sigma^{2}$.
Problem \eqref{Prob:EE_SISO_P} is a pseudo-concave maximization whose optimal solution is given by the following result. 
\begin{proposition}\label{Prop:OptP}
The global solution of Problem \eqref{Prob:EE_SISO_P} is 
\beq\label{Eq:OptP}
\bar{p}=\min(P_{max},p^{\star})\;,
\eeq
with $p^{\star}$ the unique stationary point of \eqref{Prob:EE_SISO_Pa}. 
\end{proposition}
\begin{IEEEproof}
The objective in \eqref{Prob:EE_SISO_Pa} is the ratio between a strictly concave function and an affine function, and thus it is strictly pseudo-concave. As a result, it admits a unique stationary point $p^{\star}$, which coincides with its unconstrained maximizer. Then, the thesis is obtained enforcing the constraint in \eqref{Prob:EE_SISO_Pb}, and observing that \eqref{Prob:EE_SISO_Pa} is increasing for $p\leq p^{\star}$ and decreasing for $p\geq p^{\star}$. 
\end{IEEEproof}
Thus, the global solution of \eqref{Prob:EE_SISO} can be obtained by allocating $\bGamma_{T}$ according to \eqref{Eq:OptGamma_T}, $\bGamma_{R}$ according to \eqref{Eq:OptGamma_R}, and $p$ according to \eqref{Eq:OptP}. The computational complexity required by this allocation is negligible, since closed-form results for all three radio resources have been obtained and no numerical, iterative algorithms are required. 
\begin{remark}\label{Rem:Cap_SISO}
If capacity maximization is the goal, the optimal $\bGamma_{T}$ and $\bGamma_{R}$ do not change, while the optimal transmit power is given by $p^{\star}=P_{max}$. 
\end{remark}

\section{Multiple-antenna, single-stream scenario}\label{Sec:MIMO_ST}
Let us consider now a MIMO scenario, thus relaxing the assumption $N_{R}=N_{T}=1$. However, we assume that a single data-stream is transmitted by the transmitter. This implies that the covariance matrix $\bQ$ reduces to a beamforming vector of size $N_{T}\times 1$ and squared norm $p$, i.e. $\bQ=\bq\bq^{H}$, with $\|\bq\|^{2}$. Then, the problem can be stated as 
\begin{subequations}\label{Prob:EE_SIMO}
\begin{align}
&\ds\max_{\footnotesize p,\bq,\bGamma_{T},\bGamma_{R}}\frac{B\log_{2}\left(1+\ds\frac{1}{\sigma^{2}}\left\|\bG^{H}\bGamma_{R}\bC\bGamma_{T}\bH\bq\right\|^{2}\right)}{\mu \|\bq\|^{2}+P_{c}}\label{Prob:EE_SIMOa}\\
&\;\text{s.t.}\;\|\bq\|^{2}\leq P_{max}\label{Prob:EE_SIMOb}\\
&\;\quad\;\;\|\bGamma_{T}\bH\bq\|^{2}\leq \|\bH\bq\|^{2}\label{Prob:EE_SIMOc}\\
&\;\quad\;\;\|\bGamma_{R}\bC\bGamma_{T}\bH\bq\|^{2}\leq \|\bC\bGamma_{T}\bH\bq\|^{2}\label{Prob:EE_SIMOd}
\end{align}
\end{subequations}
Proceeding similarly as in the single-antenna scenario, we will consider separately the optimization of the two WsRHSs and that of the beamforming vector. In this case, the optimization of $\bGamma_{R}$ and $\bGamma_{T}$ can be accomplished in closed-form, but the optimal $\bGamma_{R}$ and $\bGamma_{T}$ will depend on the beamforming vector $\bq$, due to the use of WsRHSs with the global reflection capabilities. Thus, it is not possible to decouple WsRHS optimization from beamforming design, as done in the single-antenna scenario, and alternating optimization between $(\bGamma_{R},\bGamma_{T})$ and $\bq$ is necessary. 

\subsection{Optimization of $\bGamma_{T}$ and $\bGamma_{R}$}\label{Sec:RHS_SIMO}
For any given $\bq$, let us redefine the vector $\bh$ as $\bh=\bH\bq$. Then, the problem reduces to
\begin{subequations}\label{Prob:EE_SIMO_H}
\begin{align}
&\ds\max_{\footnotesize \bGamma_{T},\bGamma_{R}}\left\|\bG^{H}\bGamma_{R}\bC\bGamma_{T}\bh \right\|^{2}\label{Prob:EE_SIMO_Ha}\\
&\;\text{s.t.}\;\|\bGamma_{T}\bh\|^{2}\leq \|\bh\|^{2}\label{Prob:EE_SIMO_Hb}\\
&\;\quad\;\;\|\bGamma_{R}\bC\bGamma_{T}\bh\|^{2}\leq \|\bC\bGamma_{T}\bh\|^{2}\label{Prob:EE_SIMO_Hc}
\end{align}
\end{subequations}
Defining $\bx=\bC\bGamma_{T}\bh$, $\by=\bGamma_{R}\bx$, and assuming $\bC$ has a left inverse, i.e. $\bC^{+}\bC=\bI_{M_{T}}$, Problem \eqref{Prob:EE_SIMO_H} can be restated as
\begin{subequations}\label{Prob:EE_SIMO_H2}
\begin{align}
&\ds\max_{\footnotesize \bx,\by}\left\|\bG^{H}\by\right\|^{2}\label{Prob:EE_SIMO_H2a}\\
&\;\text{s.t.}\;\|\bC^{+}\bx\|^{2}\leq \|\bh\|^{2}\label{Prob:EE_SIMO_H2b}\\
&\;\quad\;\;\|\by\|^{2}\leq \|\bx\|^{2}\label{Prob:EE_SIMO_H2c}
\end{align}
\end{subequations}
\begin{proposition}\label{Prop:OptSIMO}
The optimal solution of Problem \eqref{Prob:EE_SIMO_H2} is 
\begin{align}
\bx&=\|\bh\|\sqrt{\lambda_{max}}\bu_{max}\label{Eq:Opt_x}\\
\by&=\|\bh\|\sqrt{\lambda_{max}}\bu_{G,max}\;\label{Eq:Opt_y},
\end{align}
wherein $\lambda_{max}$ is the maximum eigenvalue of $\bC\bC^{H}$, $\bu_{max}$ the corresponding unit-norm eigenvector, and $\bu_{G,max}$ the unit-norm eigenvector corresponding to the maximum eigenvalue of $\bG\bG^{H}$.
\end{proposition}
\begin{IEEEproof}
The objective in \eqref{Prob:EE_SISO_H2a} can be written as $\by^{H}\bG\bG^{H}\by$, which implies that it is maximized when $\by=\|\by\|\bu_{G,max}$, 
with $\bu_{G,max}$ the unit-norm eigenvector of $\bG\bG^{H}$ corresponding to the maximum eigenvalue and the norm  $\|\by\|$ to be determined based on the problem constraints. Thus, proceeding similarly as in the proof of Proposition \ref{Prop:OptSISO}, the thesis follows. 
\end{IEEEproof}
Finally, from \eqref{Eq:Opt_x} and \eqref{Eq:Opt_y}, the optimal $\bGamma_{T}$ and $\bGamma_{R}$ can be obtained as in \eqref{Eq:OptGamma_T} and \eqref{Eq:OptGamma_R}.
\begin{remark}
As anticipated, in this case the optimal $\bGamma_{T}$ and $\bGamma_{R}$ depend on the beamforming vector $\bq$, since $\bq$ appears in the definition of the equivalent channel vector $\bh$. This is a direct consequence of the fact that WsRHSs with global reflection capabilities are employed, which leads to the power constraints in \eqref{Prob:EE_SIMOc} and \eqref{Prob:EE_SIMOd}, which are coupled in $\bq$ and $(\bGamma_{T},\bGamma_{R})$. As a result, the optimization of $(\bGamma_{T},\bGamma_{R})$ can not be decoupled from that of $\bq$, and alternating maximization between $(\bGamma_{T},\bGamma_{R})$ and $\bq$ will be required. 
\end{remark}

\subsection{Optimization of $\bq$}\label{Sec:TxPower_SIMO}
After optimizing $\bGamma_{T}$ and $\bGamma_{R}$, the problem becomes
\begin{subequations}\label{Prob:EE_SIMO_Beam}
\begin{align}
&\ds\max_{\footnotesize\bq}\frac{B\log_{2}\left(1+\left\|\bM\bq\right\|^{2}\right)}{\mu \|\bq\|^{2}+P_{c}}\label{Prob:EE_SIMO_Beam_a}\\
&\;\text{s.t.}\;\|\bq\|^{2}\leq P_{max}\label{Prob:EE_SIMO_Beam_b}\\
&\;\quad\;\;\|\bGamma_{T}\bH\bq\|^{2}-\|\bH\bq\|^{2}\leq 0\label{Prob:EE_SIMO_Beam_c}\\
&\;\quad\;\;\|\bGamma_{R}\bC\bGamma_{T}\bH\bq\|^{2}- \|\bC\bGamma_{T}\bH\bq\|^{2}\leq 0\label{Prob:EE_SIMO_Beam_d}\;,
\end{align}
\end{subequations}
with $\bM=\frac{1}{\sigma^{2}}\bG^{H}\bGamma_{R}\bC\bGamma_{T}\bH$. The difficulty posed by Problem \eqref{Prob:EE_SIMO_Beam} lies in the fact that the objective is not concave or pseudo-concave, and that the constraints are not convex. In order to tackle \eqref{Prob:EE_SIMO_Beam}, in the following we resort to the sequential fractional programming (SFP) framework \cite{ZapNow15}. Exploiting the fact that the norm function is convex, and so it is lower-bounded by its first-order Taylor expansion, we have that, for any vector $\bq_{0}$, it holds 
\begin{align}
\|\bM\bq\|^{2}&=\bq^{H}\bM^{H}\bM\bq\geq2\Re\{\bq_{0}^{H}\bM^{H}\bM\bq\}\notag\\
&-\bq_{0}^{H}\bM^{H}\bM\bq_{0}=2\Re\{\bbm_{0}^{H}\bq\}-m_{0}\\
\|\bH\bq\|^{2}&=\bq^{H}\bH^{H}\bH\bq\geq 2\Re\{\bq_{0}^{H}\bH^{H}\bH\bq\}\notag\\
&-\bq_{0}^{H}\bH^{H}\bH\bq_{0}=2\Re\{\bh_{0}^{H}\bq\}-h_{0}\\
\|\bC\bGamma_{T}\bH\bq\|^{2}&=\bq^{H}\bH^{H}\bGamma_{T}^{H}\bC^{H}\bC\bGamma_{T}\bH\bq\notag\\
&\hspace{-2cm}\geq 2\Re\{\bq_{0}^{H}\bH^{H}\bGamma_{T}^{H}\bC^{H}\bC\bGamma_{T}\bH\bq\}-\bq_{0}^{H}\bH^{H}\bGamma_{T}^{H}\bC^{H}\bC\bGamma_{T}\bH\bq_{0}\notag\\
&=2\Re\{\bc_{0}^{H}\bq\}-c_{0}\;,
\end{align}
wherein we have defined
\begin{align}
\bbm_{0}&\!=\!\bM^{H}\bM\bq_{0}\;,m_{0}\!=\!\bq_{0}^{H}\bM^{H}\bM\bq_{0}\\
\bh_{0}&\!=\!\bH^{H}\bH\bq_{0}\;,h_{0}\!=\!\bq_{0}^{H}\bH^{H}\bH\bq_{0}\\
\bc_{0}&\!=\!\bH^{H}\bGamma_{T}^{H}\bC^{H}\bC\bGamma_{T}\bH\bq_{0}\;,c_{0}\!=\!\bq_{0}^{H}\bH^{H}\bGamma_{T}^{H}\bC^{H}\bC\bGamma_{T}\bH\bq_{0}
\end{align}
Then, a surrogate problem for \eqref{Prob:EE_SIMO_Beam} can be stated as 
\begin{subequations}\label{Prob:EE_SIMO_Beam_LB}
\begin{align}
&\ds\max_{\footnotesize\bq}\frac{B\log_{2}\left(1+2\Re\{\bbm_{0}^{H}\bq\}-m_{0}\right)}{\mu \|\bq\|^{2}+P_{c}}\label{Prob:EE_SIMO_Beam_LBa}\\
&\;\text{s.t.}\;\|\bq\|^{2}\leq P_{max}\label{Prob:EE_SIMO_Beam_LBb}\\
&\;\quad\;\;\|\bGamma_{T}\bH\bq\|^{2}-2\Re\{\bh_{0}^{H}\bq\}+h_{0}\leq 0\label{Prob:EE_SIMO_Beam_LBc}\\
&\;\quad\;\;\|\bGamma_{R}\bC\bGamma_{T}\bH\bq\|^{2}-2\Re\{\bc_{0}^{H}\bq\}+c_{0} \leq 0\label{Prob:EE_SIMO_Beam_LBd}\;,
\end{align}
\end{subequations}
wherein the left-hand-sides of \eqref{Prob:EE_SIMO_Beam_LBc} and \eqref{Prob:EE_SIMO_Beam_LBd} are upper-bounds of the left-hand-sides of \eqref{Prob:EE_SIMO_Beam_c} and \eqref{Prob:EE_SIMO_Beam_d}, respectively. Problem \eqref{Prob:EE_SIMO_Beam_LB} is a pseudo-concave maximization subject to convex constraints, which can be solved by fractional programming \cite{ZapNow15}. Then, the SFP algorithm to tackle Problem \eqref{Prob:EE_SIMO_Beam} is stated as in Algorithm \ref{Alg:SFP_SIMO}, wherein \texttt{Obj}$(\bq)$ denotes the objective of \eqref{Prob:EE_SIMO_Beam} evaluated at $\bq$. Being an instance of the SFP method, Algorithm \eqref{Alg:SFP_SIMO} is guaranteed to monotonically improve \eqref{Prob:EE_SIMO_Beam_a} and converge to a first-order optimal point of \eqref{Prob:EE_SIMO_Beam} \cite{ZapNow15}. 
\begin{algorithm}\caption{SFP for Problem \eqref{Prob:EE_SIMO_Beam}}
\label{Alg:SFP_SIMO}
\begin{algorithmic}
\STATE $\varepsilon>0$; 
\STATE \texttt{Select any feasible} $\bq_{0}$;
\REPEAT
\STATE \texttt{Let} $\bq$ \texttt{ be the solution of Problem \eqref{Prob:EE_SIMO_Beam_LB}};
\STATE $T=\|\texttt{Obj}(\bq)-\texttt{Obj}(\bq_{0})\|$; $\bq_{0}=\bq$;
\UNTIL{$T\leq\varepsilon$}
\end{algorithmic}
\end{algorithm}

\subsection{Overall algorithm, convergence, and complexity}
The alternating maximization algorithm to tackle Problem \eqref{Prob:EE_SIMO} is formulated as in Algorithm \ref{Alg:AO_SIMO}.
\begin{algorithm}\caption{Alternating maximization for Problem \eqref{Prob:EE_SIMO}}
\label{Alg:AO_SIMO}
\begin{algorithmic}
\STATE $\varepsilon>0$; 
\STATE \texttt{Select any feasible} $\bq_{0}$, $\bGamma_{T,0}$, $\bGamma_{R,0}$;
\REPEAT
\STATE \texttt{Let} $(\bGamma_{T},\bGamma_{R})$ \texttt{ be the solution of Problem \eqref{Prob:EE_SIMO_H}, given $\bq$};
\STATE \texttt{Let} $\bq$ \texttt{be the solution of Problem \eqref{Prob:EE_SIMO_Beam}, given $(\bGamma_{T},\bGamma_{R})$}; 
\STATE $T=|\text{EE}((\bGamma_{T},\bGamma_{R}),\bq)-\text{EE}((\bGamma_{T,0},\bGamma_{R,0}),\bq_{0})|$; 
\STATE $\bGamma_{T}=\bGamma_{T,0}$, $\bGamma_{R}=\bGamma_{R,0}$, $\bq=\bq_{0}$;
\UNTIL{$T\leq\varepsilon$}
\end{algorithmic}
\end{algorithm}
\subsubsection{Convergence} Algorithm \ref{Alg:AO_SIMO} is provably convergent, as shown in the next proposition.  
\begin{proposition}\label{Prop:Convergence}
Algorithm \ref{Alg:AO_SIMO} monotonically improves the value of \eqref{Prob:EE_SIMOa} and converges in the value of the objective.
\end{proposition}
\begin{IEEEproof}
Algorithm \ref{Alg:AO_SIMO} alternatively optimizes $\bGamma_{T}$ and $\bGamma_{R}$ for fixed $\bq$, according to Proposition \ref{Prop:OptSIMO}, and $\bq$, for fixed $\bGamma_{T}$ and $\bGamma_{R}$, by running Algorithm \ref{Alg:SFP_SIMO}. Each of these two optimizations does not decrease the value of \eqref{Prob:EE_SIMOa}, because Proposition \ref{Prop:OptSIMO} provides the globally optimal $\bGamma_{T}$ and $\bGamma_{R}$ for any fixed $\bq$, while Algorithm \ref{Alg:SFP_SIMO} is an instance of the SFP framework, and thus monotonically improves the value of \eqref{Prob:EE_SIMO_Beam_a} until convergence to a point fulfilling the first-order optimality conditions of \eqref{Prob:EE_SIMO_Beam} \cite{ZapNow15}. Thus, Algorithm \ref{Alg:AO_SIMO} monotonically improves the value of \eqref{Prob:EE_SIMOa}. Moreover, the objective \eqref{Prob:EE_SIMOa} can be seen to be upper-bounded with respect to $\bGamma_{T}$, $\bGamma_{R}$, and $\bq$, because $\bGamma_{T}$ and $\bGamma_{R}$ are bounded by norm constraints, while $\|\bq\|$ appears linearly in the denominator of \eqref{Prob:EE_SIMOa} and logarithmically in the numerator of \eqref{Prob:EE_SIMOa}, thus implying that \eqref{Prob:EE_SIMOa} tends to zero for $\|\bq\|\to\infty$. Thus, the objective \eqref{Prob:EE_SIMOa} can not increase indefinitely, and Algorithm \ref{Alg:AO_SIMO} converges in the value of the objective. 
\end{IEEEproof}
\begin{remark}\label{Rem:RateMaxSS}
Algorithm \ref{Alg:AO_SIMO} can be readily specialized to perform rate maximization instead of EE maximization, by simply setting $\mu=0$. Indeed, this reduces the denominator of \eqref{Prob:EE_SIMOa} to a constant. 
\end{remark}
\subsubsection{Computational complexity}The computational complexity of Algorithm~\ref{Alg:AO_SIMO} is mostly related to Problem \eqref{Prob:EE_SIMO_Beam}, since the optimal WsRHS matrices are available in closed form for fixed $\bq$. Problem \eqref{Prob:EE_SIMO_Beam} is solved by Algorithm \ref{Alg:SFP_SIMO}, which involves the solution of a pseudo-concave problem with $N_{T}$ variables in each iteration. Next, we recall that a fractional function of $N_{T}$ variables, with concave numerator and convex denominator, can be maximized subject to convex constraints with a complexity equivalent to that of a convex problem with $N_{T}+1$ variables \cite{ZapNow15}, and that the complexity of a convex problem with $n$ variables is upper-bounded, with respect to $n$, by $n^{4}$ \cite{BenTal2001ConvexOptimization}. Thus, the complexity of Algorithm~\ref{Alg:AO_SIMO} can be upper-bounded by 
\beq
\mathcal{C}_{2}=\mathcal{O}\left(I_{1}I_{2}(N_{T}+1)^{4}\right)\;,
\eeq
with $I_{1}$ and $I_{2}$ the number of iterations for Algorithms \ref{Alg:SFP_SIMO} and \ref{Alg:AO_SIMO} to converge, respectively.  

\section{Multi-antenna, multi-stream transmission}\label{Sec:MIMO_MT}
In this section, we address the general MIMO case stated in Problem \eqref{Prob:EE_MIMO}, in which $N_{T}>1$, $N_{R}>1$, and $\text{rank}(\bQ)>1$, i.e. multiple data streams are transmitted. In this case, no closed-form solution is available and it becomes difficult to even optimize the WsRHSs $\bGamma_{R}$ and $\bGamma_{T}$ jointly. Thus, alternating maximization will be employed, iterating among $\bGamma_{R}$, $\bGamma_{T}$, $\bQ$. 

\subsection{Optimization of $\bGamma_{R}$}
For fixed $\bGamma_{T}$ and $\bQ$, define $\bB=\bC\bGamma_{T}\bH\bQ\bH^{H}\bGamma_{T}^{H}\bC^{H}$. Then, the optimization with respect to $\bGamma_{R}$ can be stated as
\begin{subequations}\label{Prob:Gamma_R_MIMO}
\begin{align}
&\ds\max_{\footnotesize \bGamma_{R}}B\log_{2}\left|\bI_{N_{R}}+\frac{1}{\sigma^{2}}\bG\bGamma_{R}\bB\bGamma_{R}^{H}\bG^{H}\right|\label{Prob:Gamma_R_MIMOa}\\
&\;\text{s.t.}\;\tr(\bGamma_{R}\bB\bGamma_{R}^{H})\leq \tr(\bB) \label{Prob:Gamma_R_MIMOb}
\end{align}
\end{subequations}
To begin with, let us focus on \eqref{Prob:Gamma_R_MIMOa}. Performing the eigenvalue decomposition $\bB=\sum_{m=1}^{M_{R}}\lambda_{m}\bu_{m}\bu_{m}^{H}$, \eqref{Prob:Gamma_R_MIMOa} becomes
\begin{align}
C&=B\log_{2}\left|\bI_{N_{R}}+\bG\bGamma_{R}\left(\sum_{m=1}^{M_{R}}\frac{\lambda_{m}}{\sigma^{2}}\bu_{m}\bu_{m}^{H}\right)\bGamma_{R}^{H}\bG^{H}\right|\notag\\
&=B\log_{2}\left|\bI_{N_{R}}+\sum_{m=1}^{M_{R}}\frac{\lambda_{m}}{\sigma^{2}}\bG\bU_{m}\bgamma_{R}\bgamma_{R}^{H}\bU_{m}^{H}\bG^{H}\right|\notag\\
&=B\log_{2}\left|\bI_{N_{R}}+\sum_{m=1}^{M_{R}}\frac{\lambda_{m}}{\sigma^{2}}\bR_{m}\bgamma_{R}\bgamma_{R}^{H}\bR_{m}^{H}\right|\;,
\end{align}
with $\bU_{m}=\text{diag}(\bu_{m})$ and $\bR_{m}=\bG\bU_{m}$, for $m=1,\ldots,M_{R}$. Next, let us rewrite the left-hand-side of \eqref{Prob:Gamma_R_MIMOb} as 
\begin{align}
&\tr(\bGamma_{R}\bB\bGamma_{R}^{H})=\tr\left(\bGamma_{R}\left(\sum_{m=1}^{M_{R}}\lambda_{m}\bu_{m}\bu_{m}^{H}\right)\bGamma_{R}^{H}\right)\\
&=\sum_{m=1}^{M_{R}}\lambda_{m}\tr\left(\bGamma_{R}\bu_{m}\bu_{m}^{H}\bGamma_{R}^{H}\right)=\sum_{m=1}^{M_{R}}\lambda_{m}\tr\left(\bU_{m}\bgamma_{R}\bgamma_{R}^{H}\bU_{m}^{H}\right)\notag\\
&=\bgamma_{R}^{H}\left(\sum_{m=1}^{M_{R}}\lambda_{m}\bU_{m}\bU_{m}^{H}\right)\bgamma_{R}^{H}=\tr\left(\bD_{R}\bgamma_{R}\bgamma_{R}^{H}\right)\;,\notag
\end{align}
with $\bD_{R}=\sum_{m=1}^{M_{R}}\lambda_{m}\bU_{m}^{H}\bU_{m}$. Then, Problem \eqref{Prob:Gamma_R_MIMO} becomes 
\begin{subequations}\label{Prob:Gamma_R_MIMO_2}
\begin{align}
&\ds\max_{\footnotesize \bGamma_{R}}\log_{2}\left|\bI_{N_{R}}+\sum_{m=1}^{M_{R}}\frac{\lambda_{m}}{\sigma^{2}}\bR_{m}\bgamma_{R}\bgamma_{R}^{H}\bR_{m}^{H}\right|\label{Prob:Gamma_R_MIMO_2a}\\
&\;\text{s.t.}\;\tr\left(\bD_{R}\bgamma_{R}\bgamma_{R}^{H}\right)\leq \tr(\bB) \label{Prob:Gamma_R_MIMO_2b}
\end{align}
\end{subequations}
Problem \eqref{Prob:Gamma_R_MIMO_2} is still challenging, due to the quadratic term $\bgamma_{R}\bgamma_{R}^{H}$. A popular approach to deal with this issue is the semi-definite relaxation technique, which defines $\widetilde{\bGamma}_{R}=\bgamma_{R}\bgamma_{R}^{H}$, and then relaxes the unit-rank constraints on $\widetilde{\bGamma}_{R}$. However, this requires a rank-reduction technique if the solution that is obtained is not unit-rank, which might cause a significant performance degradation. Another possible approach would be to reformulate the constraint\footnote{Note that $\tr(\widetilde{\bGamma}_{R})=\lambda_{max}(\widetilde{\bGamma}_{R})$ implies that $\widetilde{\bGamma}_{R}$ has only one non-zero eigenvalue since $\widetilde{\bGamma}_{R}\succeq \bzero$.} $\widetilde{\bGamma}_{R}=\bgamma_{R}\bgamma_{R}^{H}$ into $\tr(\widetilde{\bGamma}_{R})-\lambda_{max}(\widetilde{\bGamma}_{R})=0$ and $\widetilde{\bGamma}_{R}\succeq \bzero$. However, this leaves us with an equality constraint that is still non-convex, and is typically approximated by $\tr(\widetilde{\bGamma}_{R})-\lambda_{max}(\widetilde{\bGamma}_{R})\geq \epsilon$, with $\epsilon$ a small constant, and then tackled by the sequential approximation method \cite{Pan2022}. However, both methods above involve an approximation that leads to an algorithm without strong optimality claims. Instead, here we propose a different approach. We still define the matrix $\widetilde{\bGamma}_{R}=\bgamma_{R}\bgamma_{R}^{H}$, but we do not relax the unit-rank constraint and reformulate \eqref{Prob:Gamma_R_MIMO_2} as
\begin{subequations}\label{Prob:Gamma_R_MIMO_3}
\begin{align}
&\ds\max_{\footnotesize \widetilde{\bGamma}_{R}\succeq \bzero, \bgamma_{R}}\log_{2}\left|\bI_{N_{R}}+\sum_{m=1}^{M_{R}}\frac{\lambda_{m}}{\sigma^{2}}\bR_{m}\widetilde{\bGamma}_{R}\bR_{m}^{H}\right|\label{Prob:Gamma_R_MIMO_3a}\\
&\;\text{s.t.}\;\tr\left(\bD_{R}\widetilde{\bGamma}_{R}\right)\leq \tr(\bB)\label{Prob:Gamma_R_MIMO_3b} \\
&\;\quad\;\;\left[\begin{array}{ll}
\widetilde{\bGamma}_{R} & \bgamma_{R}\\
\bgamma_{R}^{H} & 1
\end{array}\right]\succeq \bzero\label{Prob:Gamma_R_MIMO_3c}\\
&\;\quad\;\;\tr(\widetilde{\bGamma}_{R})- \|\bgamma_{R}\|^{2}\leq 0\label{Prob:Gamma_R_MIMO_3d}
\end{align}
\end{subequations}
\begin{proposition}\label{Prop:Rank1_R}
Let $(\widetilde{\bGamma}_{R}^{\star},\bgamma_{R}^{\star})$ be a solution of Problem \eqref{Prob:Gamma_R_MIMO_3}. Then, $\text{rank}(\widetilde{\bGamma}_{R}^{\star})=1$.
\end{proposition}
\begin{IEEEproof}
To begin with, we prove that any feasible point of \eqref{Prob:Gamma_R_MIMO_3} is such that $\widetilde{\bGamma}_{R}-\bgamma\bgamma^{H}\succeq \bzero$. To this end, we first leverage the fact that, if $\widetilde{\bGamma}_{R}\succ \bzero$, then the matrix in \eqref{Prob:Gamma_R_MIMO_3c} is positive semi-definite if and only if its Schur's complement, i.e. $\widetilde{\bGamma}_{R}-\bgamma\bgamma^{H}$ is positive semi-definite. At this point, we proceed by contradiction to show that the result also holds when $\widetilde{\bGamma}_{R}\succeq \bzero$. Indeed, if $\lambda_{min}(\widetilde{\bGamma}_{R}-\bgamma\bgamma^{H})<0$, then, since the minimum eigenvalue is a continuous function of the elements of a matrix, it must exist $\epsilon>0$ such that 
\begin{equation}
\lambda_{min}(\epsilon\bI_{N_{R}}+\widetilde{\bGamma}_{R}-\bgamma\bgamma^{H})<0\;.
\end{equation}
This is a contradiction, since $\epsilon\bI_{N_{R}}+\widetilde{\bGamma}_{R}$ is a strictly positive definite matrix even if $\widetilde{\bGamma}_{R}\succeq 0$, and so it must hold 
\begin{equation}
\lambda_{min}(\epsilon\bI_{N_{R}}+\widetilde{\bGamma}_{R}-\bgamma\bgamma^{H})>0\;.
\end{equation}
for any $\epsilon>0$. Thus, for any feasible point of \eqref{Prob:Gamma_R_MIMO_3} it must hold that $\widetilde{\bGamma}_{R}-\bgamma\bgamma^{H}\succeq \bzero$.
Then, assuming without loss of generality that the eigenvalues of $\widetilde{\bGamma}_{R}$, say $\lambda_{1},\ldots,\lambda_{N}$, are ordered in decreasing order of magnitude, \eqref{Prob:Gamma_R_MIMO_3c} implies that 
\begin{equation}
\lambda_{1}\geq \|\bgamma\|^{2}\;,
\end{equation}
whereas \eqref{Prob:Gamma_R_MIMO_3d} requires that 
\begin{equation}
\sum_{i=1}^{N}\lambda_{i}=\lambda_{1}+\sum_{i=2}^{N}\lambda_{i}\geq \|\bgamma\|^{2}
\end{equation}
Thus, since $\lambda_{i}\geq 0$ for all $i=1,\ldots,N$, \eqref{Prob:Gamma_R_MIMO_3c} and \eqref{Prob:Gamma_R_MIMO_3d} together imply that $\lambda_{1}=\|\bgamma\|^{2}$ and $\lambda_{i}=0$ for all $i=2,\ldots,N$. Hence, the thesis. 
\end{IEEEproof}
Based on Proposition \ref{Prop:Rank1_R}, Problems \eqref{Prob:Gamma_R_MIMO_3} and \eqref{Prob:Gamma_R_MIMO_2} are equivalent. However, Problem \eqref{Prob:Gamma_R_MIMO_3} is not convex, yet, because the constraint in \eqref{Prob:Gamma_R_MIMO_3d} is not convex, since $- \|\bgamma_{R}\|^{2}$ is concave in $\bgamma_{R}$. Nevertheless, Problem \eqref{Prob:Gamma_R_MIMO_3} can be tackled by the sequential optimization framework. To this end, we need to find a convex upper-bound of the constraint function in \eqref{Prob:Gamma_R_MIMO_3d}, which can be accomplished observing that the term $\|\bgamma_{R}\|^{2}$ is convex, and thus can be lower-bounded by its first-order Taylor expansion around any feasible point $\bgamma_{R,0}$. To elaborate, for any feasible $\bgamma_{R,0}$, it holds $\|\bgamma_{R}\|^{2}\geq 2\Re\{\bgamma_{R,0}^{H}\bgamma_{R}\}-\|\bgamma_{R,0}^{H}\|^{2}$, and thus
\begin{align}
\tr(\widetilde{\bGamma}_{R})-\|\bgamma_{R}\|^{2}\leq \tr(\widetilde{\bGamma}_{R})+\|\bgamma_{R,0}^{H}\|^{2}-2\Re\{\bgamma_{R,0}^{H}\bgamma_{R}\}\;.
\end{align}
Then, Problem \eqref{Prob:Gamma_R_MIMO_3} can be replaced by the following surrogate problem, which fulfills all assumptions of the sequential optimization framework
\begin{subequations}\label{Prob:Gamma_R_MIMO_4}
\begin{align}
&\ds\max_{\footnotesize \widetilde{\bGamma}_{R}\succeq \bzero, \bgamma_{R}}\log_{2}\left|\bI_{N_{R}}+\sum_{m=1}^{M_{R}}\frac{\lambda_{m}}{\sigma^{2}}\bR_{m}\widetilde{\bGamma}_{R}\bR_{m}^{H}\right|\label{Prob:Gamma_R_MIMO_4a}\\
&\;\text{s.t.}\;\tr\left(\bD_{R}\widetilde{\bGamma}_{R}\right)\leq \tr(\bB)\label{Prob:Gamma_R_MIMO_4b} \\
&\;\quad\;\;\left[\begin{array}{ll}
\widetilde{\bGamma}_{R} & \bgamma_{R}\\
\bgamma_{R}^{H} & 1
\end{array}\right]\succeq \bzero\label{Prob:Gamma_R_MIMO_4c}\\
&\;\quad\;\;\tr(\widetilde{\bGamma}_{R})+\bgamma_{R,0}^{H}\bgamma_{R,0}-2\Re\{\bgamma_{R,0}^{H}\bgamma_{R}\}\leq 0\label{Prob:Gamma_R_MIMO_4d}\;.
\end{align}
\end{subequations}
Problem \eqref{Prob:Gamma_R_MIMO_4} is convex and thus can be solved with polynomial complexity by convex optimization methods. Then, the algorithm for Problem \eqref{Prob:Gamma_R_MIMO} can be stated as in Algorithm \ref{Alg:SFP_GammaR_MIMO}, with \texttt{Obj}$(\bgamma_{R})$ denoting \eqref{Prob:Gamma_R_MIMO_4a} evaluated at $\bgamma_{R}$. Algorithm \ref{Alg:SFP_GammaR_MIMO} is an instance of the SFP method, and thus is guaranteed to monotonically improve  \eqref{Prob:Gamma_R_MIMOa} and converges to a first-order optimal point of Problem \eqref{Prob:Gamma_R_MIMO} \cite{ZapNow15}. 
\begin{algorithm}\caption{SFP for Problem \eqref{Prob:Gamma_R_MIMO}}
\label{Alg:SFP_GammaR_MIMO}
\begin{algorithmic}
\STATE $\varepsilon>0$; 
\STATE \texttt{Select any feasible} $\bgamma_{R,0}$;
\REPEAT
\STATE \texttt{Let} $\bgamma_{R}$ \texttt{ be the solution of Problem \eqref{Prob:Gamma_R_MIMO_4}};
\STATE $T=\|\texttt{Obj}(\bgamma_{R})-\texttt{Obj}(\bgamma_{R,0})\|$; 
\STATE $\bgamma_{R,0}=\bgamma_{R}$;
\UNTIL{$T\leq\varepsilon$}
\end{algorithmic}
\end{algorithm}

\subsection{Optimization of $\bGamma_{T}$} 
$\bGamma_{T}$ can be optimized similarly as done for $\bGamma_{R}$. Defining $\bA=\bH\bQ\bH^{H}$, $\bF=\bG\bGamma_{R}\bC$, and $\bZ=\bGamma_{R}\bC$, the problem becomes
\begin{subequations}\label{Prob:Gamma_T_MIMO}
\begin{align}
&\ds\max_{\footnotesize \bGamma_{T}}B\log_{2}\left|\bI_{N_{R}}+\frac{1}{\sigma^{2}}\bF\bGamma_{T}\bA\bGamma_{T}^{H}\bF^{H}\right|\label{Prob:Gamma_T_MIMOa}\\
&\;\text{s.t.}\;\tr\left(\bGamma_{T}\bA\bGamma_{T}^{H}\right)\leq \tr(\bA) \label{Prob:Gamma_T_MIMOb}\\
&\;\quad\;\;\tr\left(\bZ\bGamma_{T}\bA\bGamma_{T}^{H}\bZ^{H}\right)-\tr\left(\bC\bGamma_{T}\bA\bGamma_{T}^{H}\bC^{H}\right)\leq 0\label{Prob:Gamma_T_MIMOc}
\end{align}
\end{subequations}
By the eigenvalue decomposition $\bA=\sum_{m=1}^{M_{T}}\beta_{m}\bv_{m}\bv_{m}^{H}$, \eqref{Prob:Gamma_T_MIMOa} becomes
\begin{align}
C&=B\log_{2}\left|\bI_{N_{R}}+\bF\bGamma_{T}\left(\textstyle\sum_{m=1}^{M_{T}}\frac{\beta_{m}}{\sigma^{2}}\bv_{m}\bv_{m}^{H}\right)\bGamma_{T}^{H}\bF^{H}\right|\notag\\
&=B\log_{2}\left|\bI_{N_{R}}+\textstyle\sum_{m=1}^{M_{T}}\frac{\beta_{m}}{\sigma^{2}}\bF\bV_{m}\bgamma_{T}\bgamma_{T}^{H}\bV_{m}^{H}\bF^{H}\right|\notag\\
&=B\log_{2}\left|\bI_{N_{R}}+\textstyle\sum_{m=1}^{M_{T}}\frac{\beta_{m}}{\sigma^{2}}\bS_{m}\bgamma_{T}\bgamma_{T}^{H}\bS_{m}^{H}\right|\;,
\end{align}
with $\bV_{m}=\text{diag}(\bv_{m})$ and $\bS_{m}=\bF\bV_{m}$, for $m=1,\ldots,M_{T}$. Next, the left-hand-side of \eqref{Prob:Gamma_T_MIMOb} can be expressed as 
\begin{align}
&\tr(\bGamma_{T}\bA\bGamma_{T}^{H})=\tr\left(\bGamma_{T}\left(\textstyle\sum_{m=1}^{M_{T}}\beta_{m}\bv_{m}\bv_{m}^{H}\right)\bGamma_{T}^{H}\right)\\
&=\textstyle\sum_{m=1}^{M_{T}}\beta_{m}\tr\left(\bGamma_{T}\bv_{m}\bv_{m}^{H}\bGamma_{T}^{H}\right)\!=\!\sum_{m=1}^{M_{T}}\beta_{m}\tr\left(\bV_{m}\bgamma_{T}\bgamma_{T}^{H}\bV_{m}^{H}\right)\notag\\
&=\bgamma_{T}^{H}\left(\textstyle\sum_{m=1}^{M_{T}}\beta_{m}\bV_{m}^{H}\bV_{m}\right)\bgamma_{T}^{H}=\tr\left(\bD_{T}\bgamma_{T}\bgamma_{T}^{H}\right)\;,\notag
\end{align}
with $\bD_{T}\!=\!\sum_{m=1}^{M_{T}}\beta_{m}\bV_{m}^{H}\bV_{m}$. Similarly, \eqref{Prob:Gamma_T_MIMOc} becomes
\begin{equation}
\tr\left(\bE_{T,1}\bgamma_{T}^{H}\bgamma_{T}\right)-\tr\left(\bE_{T,2}\bgamma_{T}^{H}\bgamma_{T}\right)\leq0\;,
\end{equation}
with $\bE_{T,1}=\sum_{m=1}^{M_{T}}\beta_{m}\bZ^{H}\bV_{m}^{H}\bV_{m}\bZ$, $\bE_{T,2}=\sum_{m=1}^{M_{T}}\beta_{m}\bC^{H}\bV_{m}^{H}\bV_{m}\bC$.
Then, Problem \eqref{Prob:Gamma_T_MIMO} becomes 
\begin{subequations}\label{Prob:Gamma_T_MIMO_2}
\begin{align}
&\ds\max_{\footnotesize \bGamma_{T}}\log_{2}\left|\bI_{N_{R}}+\textstyle\sum_{m=1}^{M_{T}}\frac{\beta_{m}}{\sigma^{2}}\bS_{m}\bgamma_{T}\bgamma_{T}^{H}\bS_{m}^{H}\right|\label{Prob:Gamma_T_MIMO_2a}\\
&\;\text{s.t.}\;\tr\left(\bD_{T}\bgamma_{T}\bgamma_{T}^{H}\right)\leq \tr(\bA) \label{Prob:Gamma_T_MIMO_2b}\\
&\;\quad\;\; \tr\left(\bE_{T,1}\bgamma_{T}\bgamma_{T}^{H}\right)-\tr\left(\bE_{T,2}\bgamma_{T}\bgamma_{T}^{H}\right)\leq 0\label{Prob:Gamma_T_MIMO_2c}
\end{align}
\end{subequations}
Defining $\widetilde{\bGamma}_{T}=\bgamma_{T}\bgamma_{T}^{H}$, Problem \eqref{Prob:Gamma_T_MIMO_2} can be reformulated as 
\begin{subequations}\label{Prob:Gamma_T_MIMO_3}
\begin{align}
&\ds\max_{\footnotesize \widetilde{\bGamma}_{T}\succeq \bzero, \bgamma_{T}}\log_{2}\left|\bI_{N_{R}}+\textstyle\sum_{m=1}^{M_{T}}\frac{\beta_{m}}{\sigma^{2}}\bS_{m}\widetilde{\bGamma}_{T}\bS_{m}^{H}\right|\label{Prob:Gamma_T_MIMO_3a}\\
&\;\text{s.t.}\;\tr\left(\bD_{T}\widetilde{\bGamma}_{T}\right)\leq \tr(\bA)\label{Prob:Gamma_T_MIMO_3b} \\
&\;\quad\;\;\tr\left(\bE_{T,1}\widetilde{\bGamma}_{T}\right)-\tr\left(\bE_{T,2}\widetilde{\bGamma}_{T}\right)\leq 0\label{Prob:Gamma_T_MIMO_3e}\\ 
&\;\quad\;\;\left[\begin{array}{ll}
\widetilde{\bGamma}_{T} & \bgamma_{T}\\
\bgamma_{T}^{H} & 1
\end{array}\right]\succeq \bzero\label{Prob:Gamma_T_MIMO_3c}\\
&\;\quad\;\;\tr(\widetilde{\bGamma}_{T})-\|\bgamma_{T}\|^{2}\leq 0\label{Prob:Gamma_T_MIMO_3d}\;.
\end{align}
\end{subequations}
\begin{proposition}\label{Prop:Rank1_T}
Let $(\widetilde{\bGamma}_{T}^{\star},\bgamma_{T}^{\star})$ be a solution of Problem \eqref{Prob:Gamma_T_MIMO_3}. Then, $\text{rank}(\widetilde{\bGamma}_{T}^{\star})=1$.
\end{proposition}
\begin{IEEEproof}
The proof follows similarly as Proposition \eqref{Prop:Rank1_R}.
\end{IEEEproof}
The non-convexity of the constraint function in \eqref{Prob:Gamma_T_MIMO_3d} can be addressed by sequential programming upon performing the first-order Taylor expansion of $\|\bgamma_{T}\|^{2}$ around any feasible $\bgamma_{T,0}$, which yields
\begin{align}
\tr(\widetilde{\bGamma}_{T})-\|\bgamma_{t}\|^{2}\leq \tr(\widetilde{\bGamma}_{T})+\|\bgamma_{T,0}\|^{2}-2\Re\{\bgamma_{T,0}^{H}\bgamma_{T}\}\;.
\end{align}
Then, the surrogate problem for the implementation of sequential programming can be obtained as
\begin{subequations}\label{Prob:Gamma_T_MIMO_4}
\begin{align}
&\ds\max_{\footnotesize \widetilde{\bGamma}_{T}\succeq \bzero, \bgamma_{T}}\log_{2}\left|\bI_{N_{R}}+\textstyle\sum_{m=1}^{M_{T}}\frac{\beta_{m}}{\sigma^{2}}\bS_{m}\widetilde{\bGamma}_{T}\bS_{m}^{H}\right|\label{Prob:Gamma_T_MIMO_4a}\\
&\;\text{s.t.}\;\tr\left(\bD_{T}\widetilde{\bGamma}_{T}\right)\leq \tr(\bA)\label{Prob:Gamma_T_MIMO_4b} \\
&\;\quad\;\;\tr\left(\bE_{T,1}\widetilde{\bGamma}_{T}\right)-\tr\left(\bE_{T,2}\widetilde{\bGamma}_{T}\right)\leq 0\label{Prob:Gamma_T_MIMO_4e}\\ 
&\;\quad\;\;\left[\begin{array}{ll}
\widetilde{\bGamma}_{T} & \bgamma_{T}\\
\bgamma_{T}^{H} & 1
\end{array}\right]\succeq \bzero\label{Prob:Gamma_T_MIMO_4c}\\
&\;\quad\;\;\tr(\widetilde{\bGamma}_{T})+\bgamma_{T,0}^{H}\bgamma_{T,0}-2\Re\{\bgamma_{T,0}^{H}\bgamma_{T}\}\leq 0\label{Prob:Gamma_R_MIMO_4d}\;,
\end{align}
\end{subequations}
which is a convex problem that can be solved by standard convex optimization tools. Thus, the sequential programming framework for  Problem \eqref{Prob:Gamma_T_MIMO} can be stated as in Algorithm \ref{Alg:SFP_GammaT_MIMO}, with \texttt{Obj}$(\bgamma_{T})$ denoting \eqref{Prob:Gamma_R_MIMO_4a} evaluated at $\bgamma_{T}$. Being an instance of the SFP method, Algorithm \ref{Alg:SFP_GammaT_MIMO} is guaranteed to monotonically improve \eqref{Prob:Gamma_T_MIMOa} and converges to a first-order optimal point of Problem \eqref{Prob:Gamma_T_MIMO} \cite{ZapNow15}. 
\begin{algorithm}\caption{SFP for Problem \eqref{Prob:Gamma_T_MIMO}}
\label{Alg:SFP_GammaT_MIMO}
\begin{algorithmic}
\STATE $\varepsilon>0$; 
\STATE \texttt{Select any feasible} $\bgamma_{T,0}$;
\REPEAT
\STATE \texttt{Let} $\bgamma_{T}$ \texttt{ be the solution of Problem \eqref{Prob:Gamma_R_MIMO_4}};
\STATE $T=\|\texttt{Obj}(\bgamma_{T})-\texttt{Obj}(\bgamma_{T,0})\|$; 
\STATE $\bgamma_{T,0}=\bgamma_{T}$;
\UNTIL{$T\leq\varepsilon$}
\end{algorithmic}
\end{algorithm}

\subsection{Optimization of $\bQ$}
Let us define $\bZ_{1}=\frac{1}{\sigma}\bG\bGamma_{R}\bC\bGamma_{T}\bH$, $\bZ_{2}=\bGamma_{R}\bC\bGamma_{T}\bH$, $\bZ_{3}=\bC\bGamma_{T}\bH$, and $\bZ_{4}=\bGamma_{T}\bH$. Then the problem to be solved can be stated as
\begin{subequations}\label{Prob:MIMO_Q}
\begin{align}
&\max_{\bQ}\frac{\log_{2}\left|\bI_{N_{R}}+\bZ_{1}\bQ\bZ_{1}^{H}\right|}{\mu\tr(\bQ)+P_{c}}\label{Prob:MIMO_Qa}\\
&\;\textrm{s.t.}\;\tr(\bZ_{4}\bQ\bZ_{4}^{H})-\tr(\bH\bQ\bH^{H})\leq 0\label{Prob:MIMO_Qb}\\
&\;\quad\;\;\tr(\bZ_{2}\bQ\bZ_{2}^{H})-\tr(\bZ_{3}\bQ\bZ_{3}^{H})\leq 0\label{Prob:MIMO_Qc}
\end{align}
\end{subequations}
Problem \eqref{Prob:MIMO_Q} is a pseudo-concave maximization subject to linear constraints. Thus, it can be globally solved with polynomial complexity by any fractional programming technique, e.g. Dinkelbach's algorithm\cite{ZapNow15}. We observe, however, that due to the presence of the global power constraints \eqref{Prob:MIMO_Qb} and \eqref{Prob:MIMO_Qc}, it is not possible to diagonalize both the objective function and constraints. Thus, Problem \eqref{Prob:MIMO_Q} must be solved directly with respect to the matrix $\bQ$.  

\subsection{Overall algorithm and computational complexity}
Finally, the alternating maximization algorithm to tackle Problem \eqref{Prob:EE_MIMO} is formulated as in Algorithm \ref{Alg:AO_MIMO}.
\begin{algorithm}\caption{Alternating maximization for Problem \eqref{Prob:EE_MIMO}}
\label{Alg:AO_MIMO}
\begin{algorithmic}
\STATE $\varepsilon>0$;  \texttt{Select any feasible} $\bQ_{0}$, $\bGamma_{T,0}$, $\bGamma_{R,0}$;
\REPEAT
\STATE \texttt{Let} $\bGamma_{R}$ \texttt{ be the output of Algorithm \ref{Alg:SFP_GammaR_MIMO}, given $\bGamma_{T}$, and $\bQ$};
\STATE \texttt{Let} $\bGamma_{T}$ \texttt{ be the output of Algorithm \ref{Alg:SFP_GammaT_MIMO}, given  $\bGamma_{R}$, and $\bQ$};
\STATE \texttt{Let} $\bq$ \texttt{be the solution of Problem \eqref{Prob:MIMO_Q}, given  $\bGamma_{R}$, and $\bGamma_{T}$}; 
\STATE $T=|\text{EE}((\bGamma_{T},\bGamma_{R}),\bQ)-\text{EE}((\bGamma_{T,0},\bGamma_{R,0}),\bQ_{0})|$; 
\STATE $\bGamma_{T}=\bGamma_{T,0}$, $\bGamma_{R}=\bGamma_{R,0}$, $\bQ=\bQ_{0}$;
\UNTIL{$T\leq\varepsilon$}
\end{algorithmic}
\end{algorithm}

\subsubsection{Convergence} Algorithm \ref{Alg:AO_MIMO} is provably convergent, as shown in the next proposition.  
\begin{proposition}\label{Prop:Convergence}
Algorithm \ref{Alg:AO_MIMO} monotonically improves the value of \eqref{Prob:aEE_MIMO} and converges in the value of the objective.
\end{proposition}
\begin{IEEEproof}
Algorithm \ref{Alg:AO_MIMO} alternatively optimizes $\bGamma_{R}$ as the output of Algorithm \ref{Alg:SFP_GammaR_MIMO},  $\bGamma_{T}$ as the output of Algorithm \ref{Alg:SFP_GammaT_MIMO}, and $\bQ$ as the solution of Problem \eqref{Prob:MIMO_Q}. Both Algorithms \ref{Alg:SFP_GammaR_MIMO} and \ref{Alg:SFP_GammaT_MIMO} are instances of the SFP method, and thus monotonically increase the value of the objective of \eqref{Prob:aEE_MIMO}. As for $\bQ$, as already mentioned, Problem \eqref{Prob:MIMO_Q} is optimally solved by standard fractional programming tools. Thus, after solving each subproblem of Algorithm \ref{Alg:AO_MIMO}, the objective in \eqref{Prob:aEE_MIMO} is not decreased. Moreover, similarly to the single-stream scenario, \eqref{Prob:aEE_MIMO} is upper-bounded. Thus it can not increase indefinitely and Algorithm \ref{Alg:AO_SIMO} converges in the value of the objective. 
\end{IEEEproof}
\begin{remark}\label{Rem:RateMax_MIMO}
Algorithm \ref{Alg:AO_MIMO} can be readily specialized to perform capacity maximization instead of EE maximization, by simply setting $\mu=0$. Indeed, this reduces the denominator of \eqref{Prob:aEE_MIMO} to a constant. 
\end{remark}
\subsubsection{Computational complexity}The computational complexity of Algorithm~\ref{Alg:AO_MIMO} is due to running Algorithms \ref{Alg:SFP_GammaR_MIMO} and \ref{Alg:SFP_GammaT_MIMO}, and solving Problem \eqref{Prob:MIMO_Q} in each iteration. Algorithm \ref{Alg:SFP_GammaR_MIMO} requires the solution of a concave problem with $M_{R}+M_{R}(M_{R}+1)/2$ variables. Thus, since the complexity of a convex problem with $n$ variables is upper-bounded, with respect to $n$, by $n^{4}$ \cite{BenTal2001ConvexOptimization}, the complexity of Algorithm \ref{Alg:SFP_GammaR_MIMO} is upper-bounded by 
\beq
{\mathcal C}_{R}=\mathcal{O}\left(I_{3}\left(M_{R}+\frac{M_{R}(M_{R}+1)}{2}\right)^{4}\right)\;,
\eeq
with $I_{3}$ the number of iterations for Algorithm \ref{Alg:SFP_GammaR_MIMO} to converge. Similarly, Algorithm \ref{Alg:SFP_GammaT_MIMO} requires the solution of a concave problem with $M_{T}+M_{T}(M_{T}+1)/2$, and, thus, its complexity can be upper-bounded by 
\beq
{\mathcal C}_{T}=\mathcal{O}\left(I_{4}\left(M_{T}+\frac{M_{T}(M_{T}+1)}{2}\right)^{4}\right)\;,
\eeq
with $I_{4}$ the number of iterations for Algorithm \ref{Alg:SFP_GammaT_MIMO} to converge. Instead, Problem \eqref{Prob:MIMO_Q} is a fractional maximization whose objective has a concave numerator and a convex denominator. Thus, recalling that a fractional function with $m$ variables, concave numerator, and convex denominator, can be maximized subject to convex constraint with a complexity equivalent to that of a convex problem with $m+1$ variables \cite{ZapNow15}, the complexity of Problem \eqref{Prob:MIMO_Q} can be upper-bounded by 
\beq
{\mathcal C}_{Q}=\mathcal{O}\left(\frac{N_{T}(N_{T}+1)}{2}^{4}\right)
\eeq
Finally, the overall complexity of Algorithm  \ref{Alg:AO_MIMO} can be upper-bounded by 
\beq
{\mathcal C}_{5}=I_{5}\left({\mathcal C}_{R}+{\mathcal C}_{T}+{\mathcal C}_{Q}\right)\;,
\eeq
with $I_{5}$ the number of iterations for Algorithm  \ref{Alg:AO_MIMO} to converge.

\section{Numerical Analysis} \label{Sec:NUM_ANA}
For our numerical analysis, we consider the setup described in Section~\ref{Sec:SysModel}. The carrier frequency is $f_{c}=3.5\,\textrm{GHz}$ and the bandwidth $B=20\,\textrm{MHz}$. The transmitter and receiver are $100\,\textrm{m}$ apart and have $N_{T}=4$ and $N_{R}=4$ antennas. The transmit and receive WsRHS have $M_{T}=100$ and $M_{R}=100$ elements. The channels matrices $\bH$ and $\bG$ follow the spherical wave model from Section~\ref{Sec:SysModel}, while the entries of the channel matrix $\bC$ follow a Rice fading model with Rice factor $10$. Moreover, $P_{c,RHS}=0\,\textrm{dBm}$, $P_{c,a}=34\,\textrm{dBm}$, while  $P_{c,0}=37\,\textrm{dBm}$, $\mu=1$. The noise power spectral density is $-174\,\textrm{dBm/Hz}$, and the receive noise figure is $5\,\textrm{dB}$.  The first set of figures, namely Figs. ~\ref{fig:EE}-\ref{fig:EH}, consider the case in which a single data-stream is transmitted. This applies, for example, to a cellular communication in which either the transmitter or the receiver have a single antenna, as it is typically the case for mobile terminals in a cellular network. Next, Figs.~\ref{fig:EM}-\ref{fig:EE_000} present numerical results for the case of multi-stream transmission. 

Fig.~\ref{fig:EE} shows the EE versus the maximum available transmit power $P_{max}$ obtained by the following resource allocation policies:
\begin{itemize}
\item EE maximization by Algorithm \ref{Alg:AO_SIMO}.
\item Capacity maximization by specializing Algorithm \ref{Alg:AO_SIMO} as described in Remark \ref{Rem:RateMaxSS}.
\item EE maximization assuming $N_{T}=N_{R}=1$, as described in Section \ref{Sec:SISO}.
\item Capacity maximization with $N_{T}=N_{R}=1$, specializing the approach from Section \ref{Sec:SISO} as described in Remark \ref{Rem:Cap_SISO}.
\item EE maximization assuming the two WsRHSs are not present, and $N_{T}=N_{R}=64$. 
\item Capacity maximization assuming the two WsRHSs are not present, and $N_{T}=N_{R}=64$. 
\end{itemize}
The results indicate that the highest EE value is obtained when a single antenna is used at both the transmitter and receiver, followed closely by the scenario in which $N_{T}=N_{R}=4$. Instead, when the size of the digital antenna array is increased to $N_{T}=N_{R}=64$, a significant decrease of the EE is observed. This is expected, since the power consumption of a digital radio-frequency chain is higher than a single reflecting element of the WsRHSs. 
\begin{figure}[!h]
	\centering
	\includegraphics[width=0.5\textwidth]{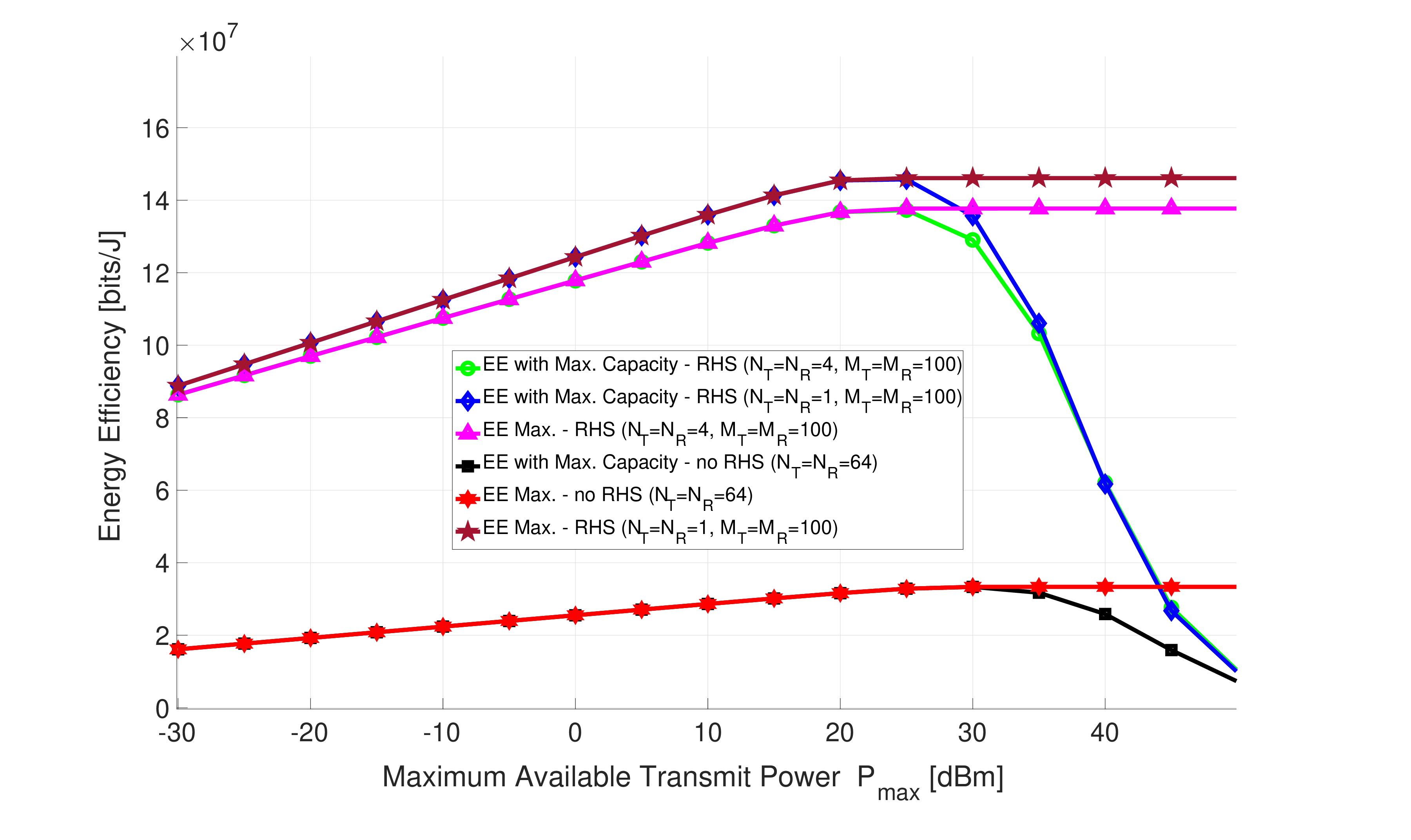}\caption{Achieved EE versus Maximum Available Transmit Power $P_{max}$ for Single Stream} \label{fig:EE}
\end{figure}
The same scenarios are considered in Fig.~\ref{fig:SE}, but the reported metric is the achieved capacity rather than the EE. The results show that using two WsRHSs provides satisfactory capacity levels, even though the higher capacity is obtained when $N_{T}=N_{R}=4$ and not when a single-antenna is used. On the other hand, using a large number of radio-frequency chains does not lead to better capacity values. This is explained recalling that the considered scenario assumes a single-stream transmission. Thus, no multiplexing gain is obtained by increasing the values of $N_{T}$ and $N_{R}$, but only an array gain, which is also provided by the WsRHSs. 

\begin{figure}[!h]
	\centering
	\includegraphics[width=0.5\textwidth]{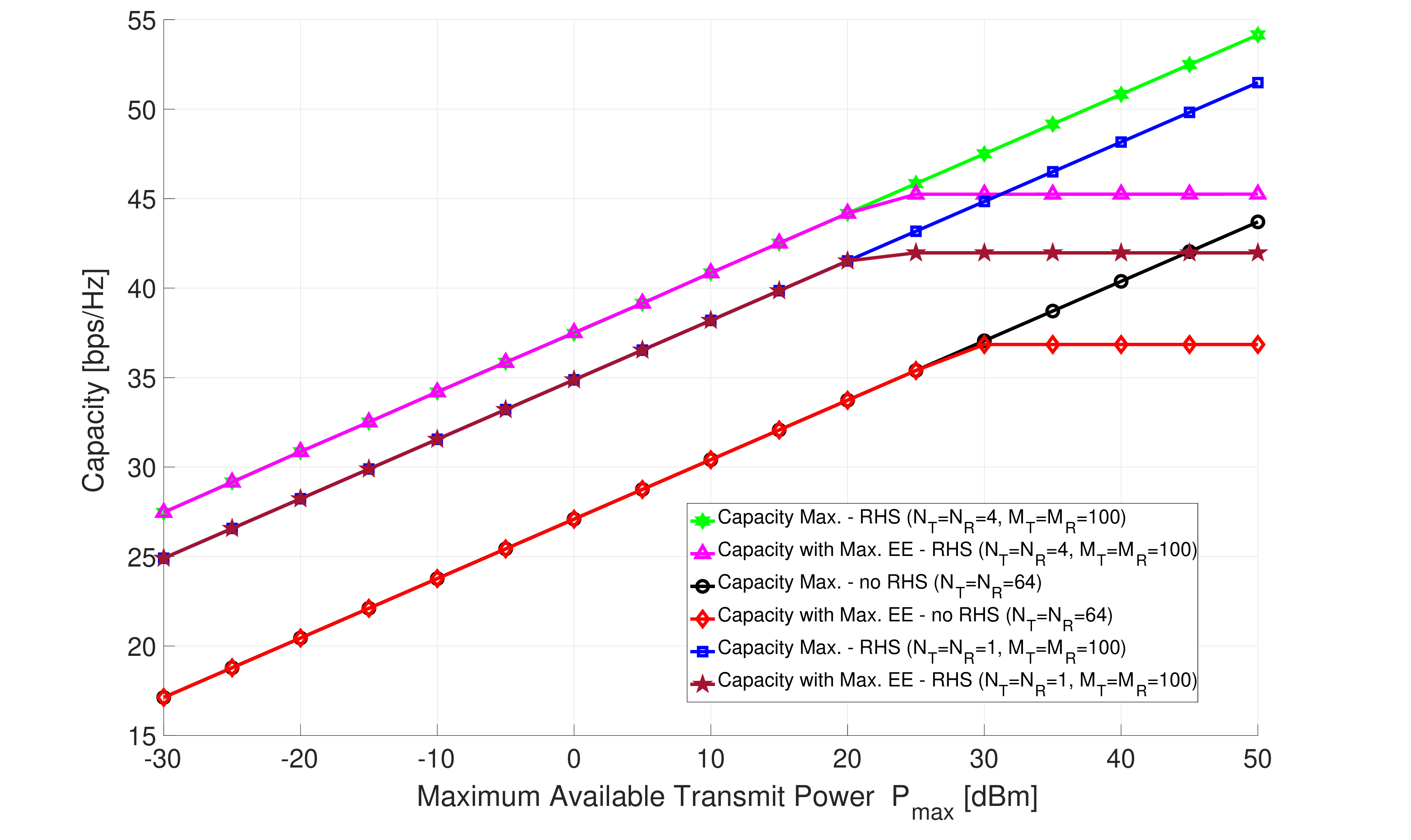}\caption{Capacity versus Maximum Available Transmit Power $P_{max}$ for Single Stream.} \label{fig:SE}
\end{figure}

Next, Fig.~\ref{fig:EH} shows the EE obtained by Algorithm \ref{Alg:AO_SIMO} with $M_{T}=M_{R}=100$ and $N_{T}=N_{R}=1,4$, versus the hardware power consumption per radio-frequency antenna chain, considering the single-stream scenario, and compares it to the EE achieved when no WsRHSs are used, but $N_{T}=N_{R}=32,64,128$. The results confirm that the use of WsRHSs leads to a significant EE improvement compared to the use of a large number of digital antennas. This result was expected since the use of a single-stream transmission causes the fact that all considered beamforming schemes have a multiplexing gain of $1$, and thus only an array gain can be provided by the system. Moreover, it is seen that, for low values of hardware power consumption, using four digital antennas outperforms the use of a single digital antenna. However, the situation changes as the hardware power consumption per antenna increases, and eventually, using a single digital antenna provides the best EE level. 

\begin{figure}[!h]
	\centering
	\includegraphics[width=0.5\textwidth]{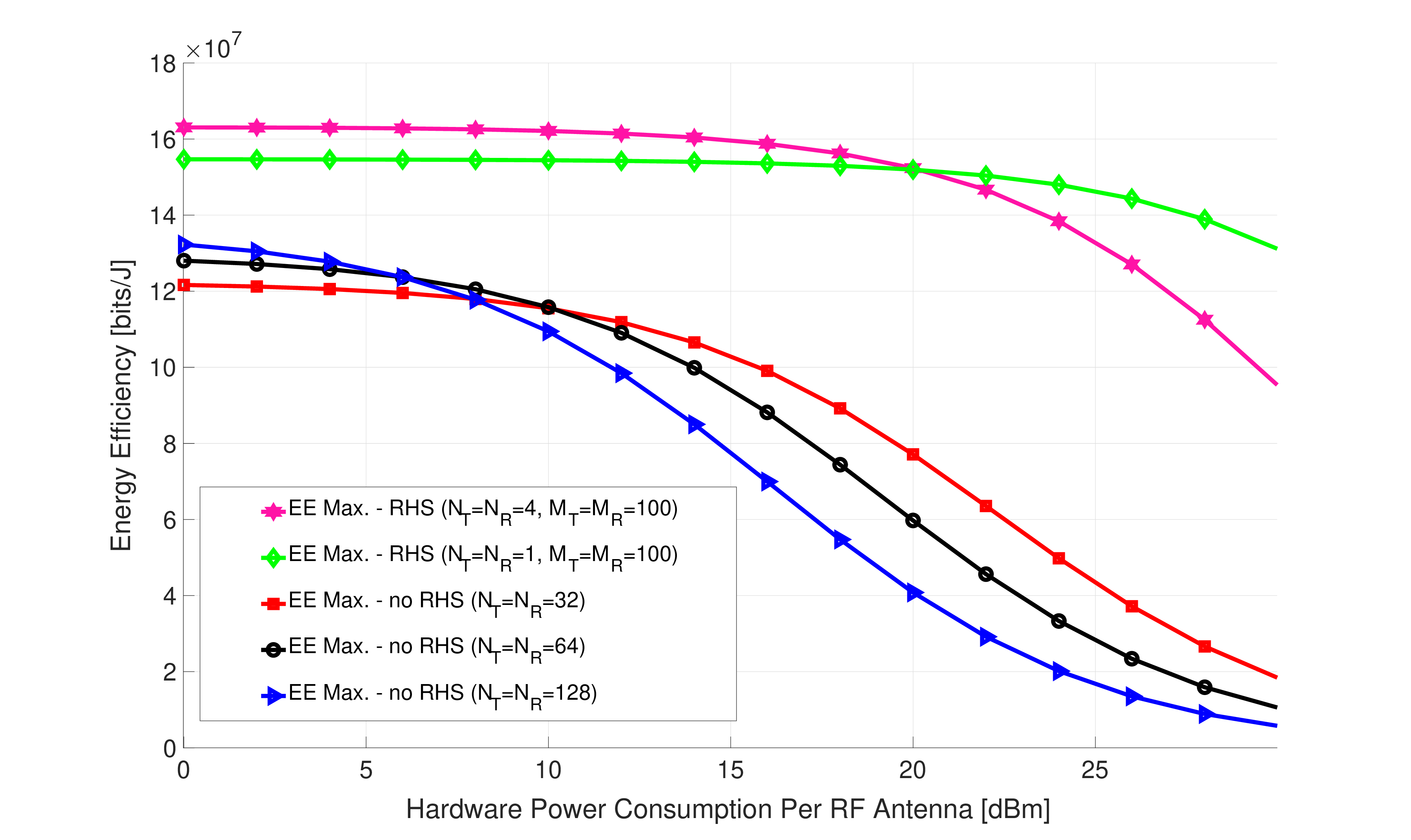}\caption{Achieved EE versus Hardware Power Consumption per RF Antenna for Single Stream.} \label{fig:EH}
\end{figure}

So far, we have presented results in which a data single-stream is transmitted, which means that no multiplexing gain is possible. As mentioned, this is always the case if either the transmitter or the receiver is a mobile user with a single-antenna, which is the typical assumption in massive MIMO networks. On the other hand, in several practical scenarios, both the transmitter and the receiver may be equipped with multiple antennas. In this case, a multiplexing gain is possible, and it is interesting to compare the EE and capacity of the system with and without WsRHSs, for different number of transmit and receive antennas. Specifically, Fig.~\ref{fig:EM} presents the achieved EE versus the maximum available transmit power $P_{max}$ for multi-stream transmissions, considering the following scenarios:
\begin{itemize}
\item Maximization of the EE by Algorithm \ref{Alg:AO_MIMO}, with $N_{T}=N_{R}=2$, and $M_{T}=M_{R}=64$.  
\item Maximization of the capacity by specializing Algorithm \ref{Alg:AO_MIMO} as described in Remark \ref{Rem:RateMax_MIMO}.
\item Maximization of the EE assuming the two WsRHSs are not present, and $N_{T}=N_{R}=4, 8, 16$. 
\item Maximization of the capacity assuming the two WsRHSs are not present, and $N_{T}=N_{R}=4, 8, 16$. 
\end{itemize}
The results indicate that using WsRHSs provides the highest EE, even if the system with WsRHSs has a multiplexing gain of $2$, while using a larger number of transmit/receive antennas provides a larger multiplexing gain, equal to $4$, $8$, or $16$. It is also seen that the EE gain diminishes for higher values of $P_{max}$, and the configuration with $16$ antennas performs similarly to the system with WsRHSs. This result indicates that while digital architectures can perform competitively at high transmit power, they appear to incur higher energy consumption due to the large number of RF chains, especially at low and moderate transmit power levels, despite the presence of a multiplexing gain.  Overall, these results confirm that the use of WsRHS at both the transmitter and receiver leads to significantly better EE in low/moderate power regimes, and can maintain competitive EE even as the transmit power increases.

\begin{figure}[!h]
	\centering
	\includegraphics[width=0.5\textwidth]
    {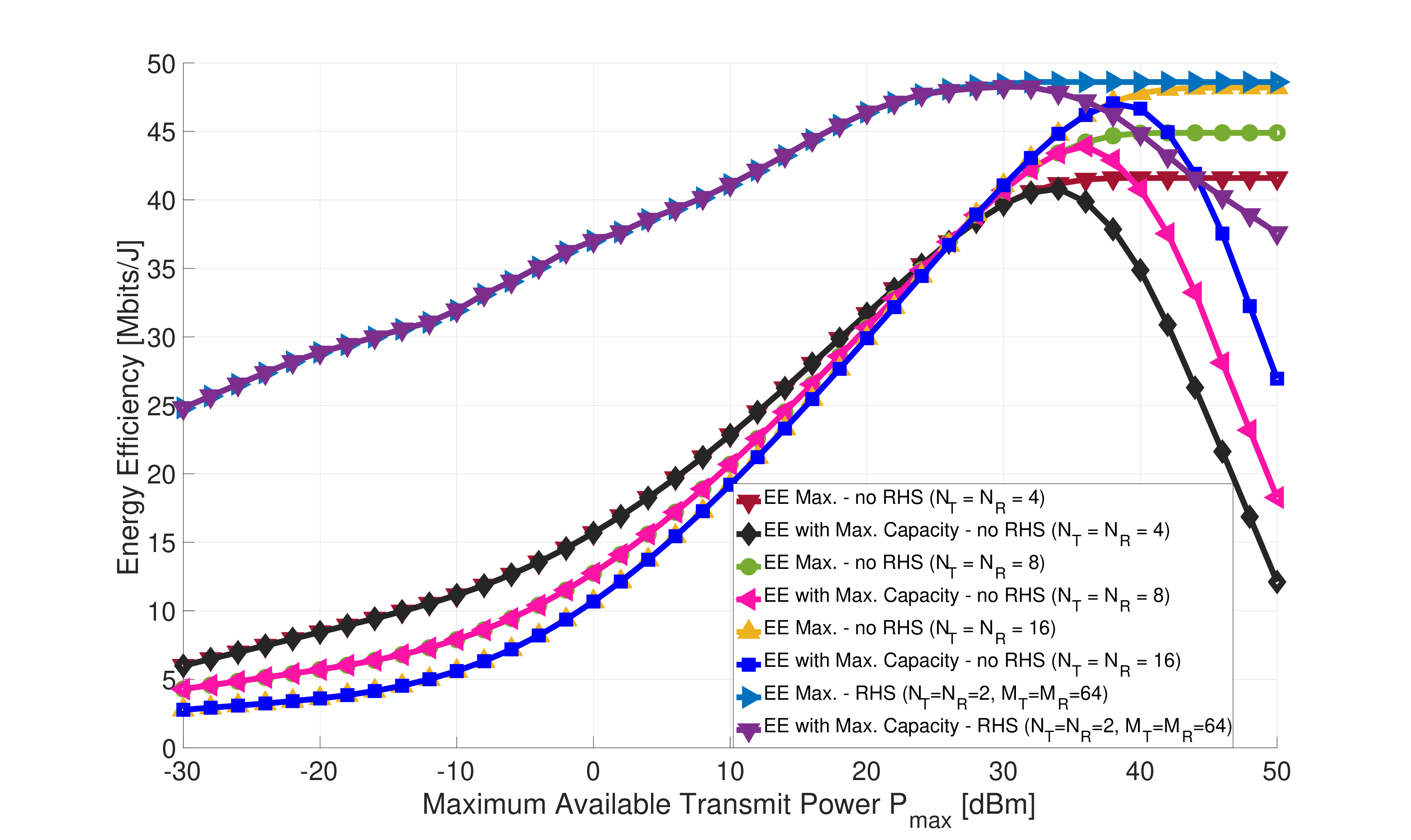}\caption{Achieved EE versus Maximum Available Transmit Power $P_{max}$ for Multiple Stream} \label{fig:EM}
\end{figure}

Fig.~\ref{fig:SM} considers the same resource allocation schemes as  Fig.~\ref{fig:EM}, but shows the achieved capacity instead of EE. In this case, the system with WsRHSs provides better capacity levels than the fully digital architectures only at very low power levels. Instead, at higher power levels, the WsRHS-based architecture suffers a significant gap compared to the fully digital architectures, due to the higher multiplexing gain provided by the use of a larger number of transmit/receive antennas. This results was expected since the energy cost required to provide a higher multiplexing gain is not accounted if capacity is the goal of the resource allocation. Thus, if capacity maximization is the objective of the resource allocation process, without any concerns for EE or architecture complexity, then digital beamforming provides better results.  

\begin{figure}[!h]
	\centering
	\includegraphics[width=0.5\textwidth]
    {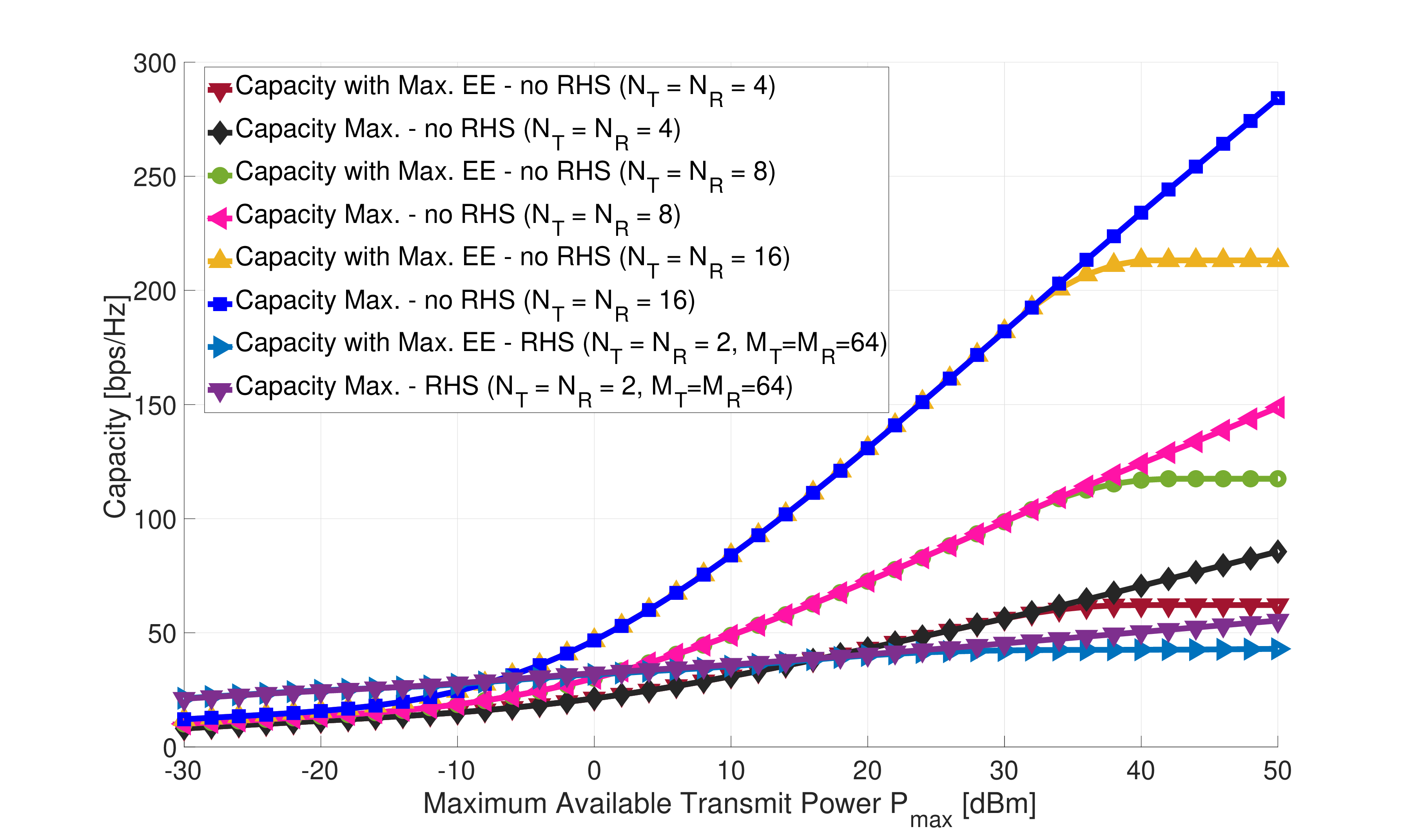}\caption{Capacity versus Maximum Available Transmit Power $P_{max}$ for Multiple Stream} \label{fig:SM}
\end{figure}

%\begin{figure}[!h]
%	\centering
%	\includegraphics[width=0.53\textwidth]{EE_Vs_Static_Power_at_-10dBm.eps}\caption{Achieved EE versus Hardware Power Consumption per RF Antenna for Multiple Stream at $P_{max}=-10 dBm$} \label{fig:EH_-10}
%\end{figure}

%\begin{figure}[!h]
%	\centering
%	\includegraphics[width=0.53\textwidth]{EE_Vs_Static_Power_at_0dBm.eps}\caption{Achieved EE versus Hardware Power Consumption per RF Antenna for Multiple Stream at $P_{max}=0 dBm$} 
 % \label{fig:EH_0}
%\end{figure}

%\begin{figure}[!h]
%	\centering
%	\includegraphics[width=0.52\textwidth]{EE_Vs_Static_RF_Power.eps}\caption{Achieved EE versus Hardware Power Consumption per RF Antenna for Multiple Stream at $P_{max}=30 dBm$} \label{fig:EH_10}
%\end{figure}

Lastly, Fig.~\ref{fig:EE_000} parallels Fig. \ref{fig:EH} for the multi-stream scenario. Specifically, it shows the EE obtained by Algorithm \ref{Alg:AO_MIMO} with $M_{T}=M_{R}=64$ and $N_{T}=N_{R}=2$ versus the hardware power consumption per RF antenna, and compares it to the EE achieved when no WsRHSs are used, but $N_{T}=N_{R}=4,8,16$. The comparison is made for three values of $P_{max}$, namely $P_{max}=-10,0,30\,\textrm{dBm}$. The results show that, for $P_{max}=-10\,\textrm{dBm}$, the use of WsRHSs provides much higher EE compared to the use of digital beamforming, despite the fact that, in the multi-stream scenario, a multiplexing gain equal to the number of digital antennas is present. Thus, for low values of $P_{max}$, even when $N_{T}=N_{R}=16$, i.e. with a multiplexing gain of $16$, digital beamforming provides a lower EE than that obtained by the WsRHS-based solution, in which $N_{T}=N_{R}=2$, i.e. with a multiplexing gain of $2$, over the complete range of hardware power that was considered. Instead, when $P_{max}=0\,\textrm{dBm}$, the EE of the WsRHS-based solution is still higher than that obtained by digital beam-forming for $N_{T}=N_{R}=4,8$, but a crossing point is observed when $N_{T}=N_{R}=16$. Thus, for intermediate values of $P_{max}$, the WsRHS-based solution outperforms fully digital beamforming when the per-antenna power consumption is limited. Finally, for large values of $P_{max}$, e.g. $P_{max}=30\,\textrm{dBm}$, the WsRHS-based solution is outperformed by digital beamforming, unless the per-antenna hardware power consumption is very large. 

\begin{figure*}[tb]
  \begin{center}
    \begin{tabular}{c}
        \begin{minipage}{0.33\hsize}
        \begin{center}
          \includegraphics[width=1.0\linewidth]{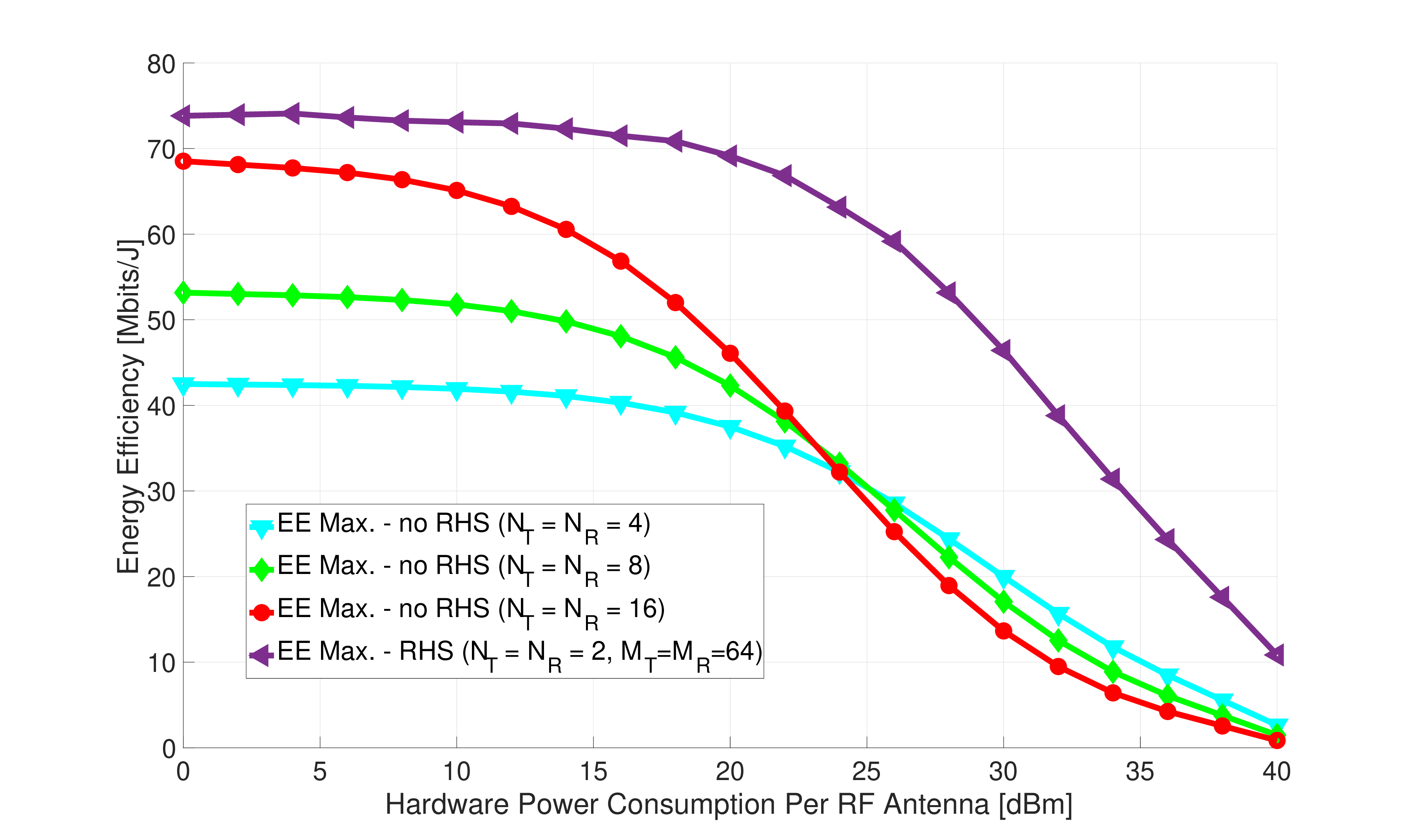}
          \hspace{0.1cm} (a) $P_{max}=-10\,\textrm{dBm}$
        \end{center}
      \end{minipage}

      \begin{minipage}{0.33\hsize}
        \begin{center}
          \includegraphics[width=1.0\linewidth]{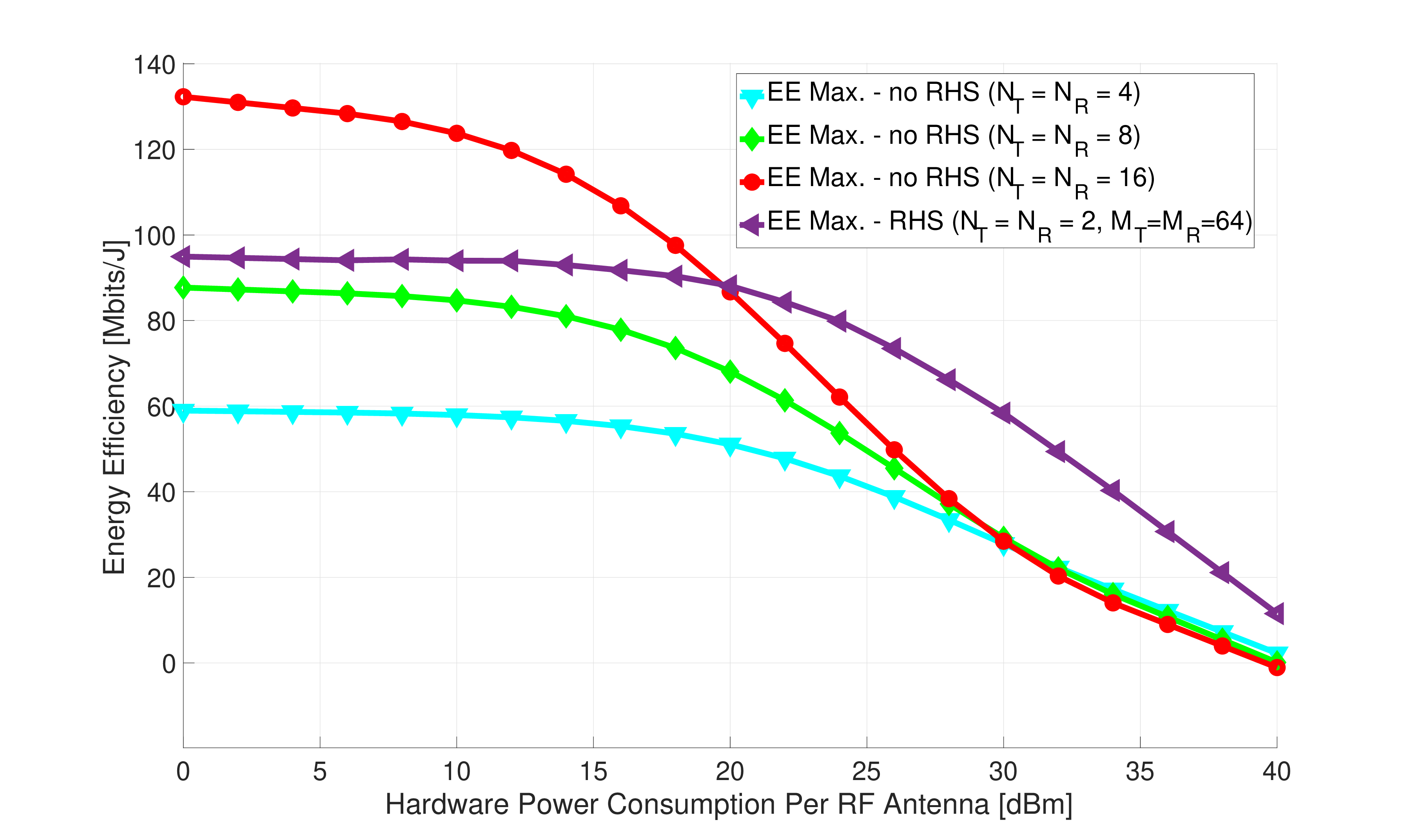}
          \hspace{0.1cm} (b) $P_{max}=0\,\textrm{dBm}$
        \end{center}
      \end{minipage}

      \begin{minipage}{0.33\hsize}
        \begin{center}
          \includegraphics[width=1.0\linewidth]{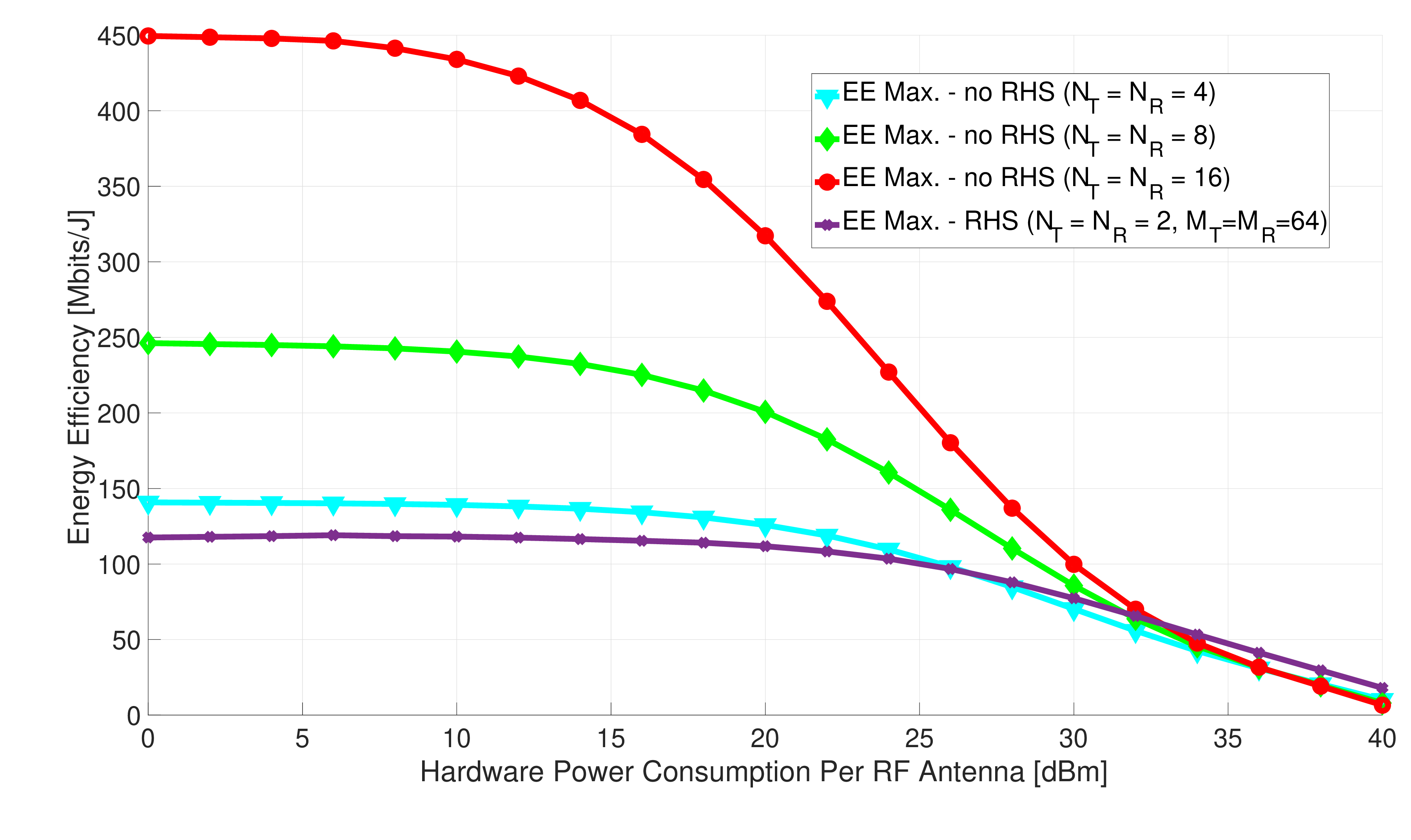}
          \hspace{0.1cm} (c) $P_{max}=30\,\textrm{dBm}$
        \end{center}
      \end{minipage}

      \begin{minipage}{0.06\hsize}
        \vspace{5mm}
      \end{minipage}
    \end{tabular}
    \caption{Achieved EE versus Hardware Power Consumption per RF Antenna for Multiple Streams.}
    \label{fig:EE_000}
  \end{center}
\end{figure*}

\section{Conclusion} \label{Sec:Concl}
\noindent This work has considered the problem of EE maximization in a MIMO wireless link aided by WsRHSs with global reflection constraints. A closed-form solution to the EE maximization problem has been obtained for the single-antenna scenario, while provably convergent numerical algorithms have been provided in the multiple-antenna case, with single-stream and multi-stream transmissions. The analysis shows that, in the single-stream scenario, the use of WsRHSs allows significantly reducing the number of digital antennas, while providing a large EE gain. On the other hand, when a multiplexing gain is possible and multiple data-streams are transmitted, the WsRHS-based solution with a low number of digital antennas achieves a significant EE gain for low and intermediate values of transmit powers, despite the lower multiplexing gain. In this context, a future line of investigation is the consideration of system optimization in the electromagnetic domain \cite{Wasif25}, which, however, leads to more complicated system models. 
 
\bibliographystyle{IEEEtran}
\bibliography{Biblio}

\end{document}

%% file: abbreviations.tex
\newcommand{\eeq}{\end{equation}}

\newcommand{\br}{\bm{r}}

\newcommand{\bs}{\bm{s}}

\newcommand{\bq}{\bm{q}}

\newcommand{\bzero}{\bm{0}}
\newcommand{\bM}{\bm{M}}

\newcommand{\bA}{\bm{A}}

\newcommand{\bC}{\bm{C}}
\newcommand{\bX}{\bm{X}}
\newcommand{\bh}{\bm{h}}
\newcommand{\bH}{\bm{H}}

\newcommand{\bc}{\bm{c}}
\newcommand{\bS}{\bm{S}}
\newcommand{\bgamma}{\bm{\gamma}}

\newcommand{\bZ}{\bm{Z}}
\newcommand{\bU}{\bm{U}}
\newcommand{\bGamma}{\bm{\Gamma}}
\newcommand{\bB}{\bm{B}}
\newcommand{\bE}{\bm{E}}
\newcommand{\bV}{\bm{V}}

\newcommand{\bG}{\bm{G}}

\newcommand{\bu}{\bm{u}}

\newcommand{\bx}{\bm{x}}

\newcommand{\bg}{\bm{g}}
\newcommand{\bn}{\bm{n}}
\newcommand{\bbm}{\bm{m}}
\newcommand{\bv}{\bm{v}}

\newcommand{\bQ}{\bm{Q}}

\newcommand{\bR}{\bm{R}}
\newcommand{\bF}{\bm{F}}

\newcommand{\bI}{\bm{I}}

\newcommand{\bD}{\bm{D}}
\newcommand{\by}{\bm{y}}
\newcommand{\ds}{\displaystyle}

\newcommand{\beq}{\begin{equation}}

\newcommand{\tr}      {{\mathrm{tr}}}    % trace of a matrix
\newtheorem{remark}{Remark}

\newtheorem{proposition}{Proposition}